\documentclass{article}
\usepackage{fullpage}
\usepackage{amsmath}
\usepackage{amsfonts}
\usepackage{amssymb}
\usepackage{epsf}
\usepackage{graphics}
\usepackage{float}
 \usepackage{graphicx}
 \usepackage{color}

\usepackage{mathrsfs} 
\usepackage{breakcites}
\usepackage{breqn}

\restylefloat{figure}
\usepackage{amsthm}
\usepackage{rotating}
\usepackage{subfigure}   
\usepackage{verbatim}
\usepackage[printonlyused]{acronym} 
\usepackage{setspace}    
\definecolor{DarkGreen}{rgb}{0.1,0.6,0.0}

\newcommand{\vect}[1]{{\bf {#1} }}

\usepackage{multimedia}

\theoremstyle{plain}
\newtheorem{theorem}{Theorem}

\newtheorem{lemma}[theorem]{Lemma}

\newcommand{\half}{\frac{1}{2}}

\begin{document}

\bibliographystyle{apalike}    

\title{Bathymetry and friction estimation from transient velocity data for 1D shallow water flows in open channels with varying width}


\author{Miguel Angel Moreles\footnote{Miguel Angel Moreles, Centro de Investigaci\'{o}n en Matem\'{a}ticas,  Jalisco s/n, Valenciana, Guanajuato, Gto 36240, Mexico, moreles@cimat.mx},
Gerardo Hernandez-Duenas \footnote{ Gerardo Hernandez-Duenas, Instituto de Matem\'aticas - Juriquilla, Universidad Nacional Aut\'onoma de M\'exico, Blvd. Juriquilla 3001, Quer\'etaro, 76230, M\'exico, hernandez@im.unam.mx} \footnote{Corresponding Author} ,
Pedro Gonzalez-Casanova  \footnote{Instituto de Matem\'aticas - CdMx, Universidad Nacional Aut\'onoma de M\'exico, \'Area de la Investigaci\'on Cient\'ifica, Circuito exterior, Ciudad Universitaria, 04510, M\'exico, CDMX, casanova@im.unam.mx}
}
 
\maketitle

\thanks{Research supported in part by grants UNAM-DGAPA-PAPIIT  IN113019 \& Conacyt A1-S-17634.}\\

\thispagestyle{empty}

\begin{abstract}
The shallow water equations (SWE) model a variety of geophysical flows. Flows in channels with rectangular cross sections may be modelled with a simplified one-dimensional SWE with varying width. Among other model parameters, information about the bathymetry and friction coefficient is needed for the correct and precise prediction of the flow. Although synthetic values of the model parameters may suffice for testing numerical schemes, approximations of the bathymetry and other parameters may be required for applications. Estimations may be obtained by experimental methods but some of those techniques may be expensive, time consuming, and not always available. In this work, we propose to solve the inverse problem to estimate the bathymetry and the Manning's friction coefficient from transient velocity data. This is done with the aid of a cost functional which includes the SWE through Lagrange multipliers. The solution is obtained by solving the constrained optimization problem by a continuous descent method. The direct and the adjoint problems are both solved numerically using a second-order accurate Roe-type upwind scheme. Numerical tests are included to show the merits of the algorithm.
\end{abstract}


\section{Introduction}

The shallow water equations (SWE), also known as the Saint-Venant system, model a variety of geophysical flows and a vast amount of applications exists in the literature. See for instance,  \cite{leveque2011tsunami}, where numerical challenges in modelling inundation of small-scale coastal regions were analyzed. Hydrodynamic modelling of open-channel flows involves the solution of 1D shallow water systems \cite{vazquez1999improved}. See \cite{bellos1992experimental,khan2014modeling} for a list of experiments in channels with different conditions. The case of 
channels with vertical walls and variable cross-sectional width is studied in \cite{balbas2009central}.

Direct problems for the above phenomena have been intensively studied during the last decades. Computing the corresponding solutions requires the knowledge of the bed channel bathymetry, appropriate initial and boundary conditions  and possibly specific model parameters such as friction's coefficients. The computation of the topography or bathymetry, needed specially in applications, may be obtained with the use of experimental methods. Examples of these techniques for river bathymetry include  interferometric synthetic aperture radar (SAR), digital photogrammetry \cite{Marks_K,Westaway} or airborne laser altimetry (LiDAR). Although these methods are the most widely used for these problems, these techniques are expensive and time consuming, and are not always available.

In this work, we focus on the signature that the bathymetry leaves on the perturbations in transient flows given by the shallow water equations in channels with vertical walls and uniform width. This leads to an alternative approach. Namely, the solution of an inverse problem to estimate the channel's bathymetry and the Manning's friction coefficient from appropriate measurements of available data. 

Research on this area is active. Let us discuss some recent contributions  relevant to the present work such as those related to hyperbolic PDE-constrained problems and specifically the shallow water equations in channels.  An optimal estimation for parameters in 1-D hyperbolic systems based on the adjoint method was proposed in \cite{Nguyen_etal_2016}. Applications to parameter estimation of a real hydrological system or overland flow with infiltration can be found in \cite{nguyen2016parameter} and \cite{nguyen2014optimal}, respectively. The problem of optimally determining the initial conditions in the shallow water equations was treated in \cite{Kevlahan_etal}, giving sufficient conditions for convergence to the true initial conditions. Monte-Carlo and gradient-based optimization methods used to calibrate the shallow water equations are studied in \cite{Lacasta_etal}. 

Not only the topography in hydrological systems can be estimated. The Manning roughness's coefficient has also been identified in the context of a complex channel network in \cite{ding2004identification,ding2005identification,ding2012optimal}. Furthermore, a general framework to deal with hyperbolic PDE_constrained optimization problems was presented in \cite{montecinos2019numerical}. A coupled system of the PDE-constrained problem and the adjoint formulation was presented, and conditions are provided to guarantee existence of an optimal solution. A direct numerical approach to reconstruct river bed topography from free surface elevation data is presented in \cite{Gessese_etal_2011}. Then, in \cite{Gessese_etal_2013} the authors consider velocity measurements and assume a steady flow. Bathymetry imaging using depth-averaged quasi-steady velocity observations are carried out in \cite{Lee_etal}.



Optimal flood control in rivers and watersheds have been successfully investigated in the literature. Among them, adjoint sensitivity analysis (ASA) based on a variational principle has been applied to find the sensitivity of hydrodynamic variables with respect to control variables in one- and two- dimensional flow models \cite{Kawahara,Sandersa,Sandersb,Sandersc,Ding}. Variational data assimilation (VDA) and adjoint sensitivity analysis (ASA) have been used for estimating unknown bathymetries of rivers and improving flood predictions by assimilating observed flow variables from measurements and satellite images. See \cite{Mazauric2003,Dimet2003,MAZAURIC2004403} for more details. A similar variational approach for Lagrangian data assimilation in rivers applied to the identification of bottom topography and initial water depth and velocity estimation was presented in \cite{Honnorata,Honnoratb,Honnoratc}). Specifically in \cite{Honnoratb}, taking into account that the velocity measurements can be scarce and that they usually require complex human interventions, the authors introduce a method which uses Lagrangian data, that can be extracted from video images, into the assimilation process. This method is applied by the authors to the identification of bed elevation and initial conditions for an academic configuration of a river hydraulic model. Such model is based on the shallow water equations. 


In \cite{P2Languer2005}, a four-dimensional variational data assimilation technique is presented as a tool to forecast floods. They modified the shallow-water equations to include a simplified sediment transport model. Their main purpose is to compute the thickness of the sediment layer bounded by the height of the solid bathymetry or bedrock and the mobile bathymetry composed of sediments.

We finally note that in \cite{P4Brisset2018}, given altimetry measurements, the identification capability of time varying inflow discharge and the Strickler coefficient (defined as a power-law in the water depth) of the 1D river Saint--Venant model is investigated. The bathymetry, however, is either provided or estimated from one in-situ measurement following \cite{GARAMBOIS2015103} and  \cite{Gessese_etal_2011}. Estimations of the bathymetry from discrete surface velocity data is currently an important and active field of research.


In this work, we are concerned with the next natural case of practical interest. Namely, the one of recovering the bed bathymetry and the Manning's friction coefficient from prescribed transient velocity data in open channels with vertical walls and varying width.  Estimating the bathymetry and the Manning's friction coefficient from transient velocity measurements is a lot more challenging compared to a situation where the flow corresponds to a steady state. To our knowledge this problem has not been addressed in the literature in the case of channels with varying width. More precisely, in this work, using a cost functional subject to the shallow water equations, we formulate an algorithm capable of recovering the bathymetry from a set of point values of the fluid's velocity in the channel. Specifically, we formulate a cost functional which includes the SWE through Lagrange multipliers. Boundary and initial conditions are included. The constrained optimization problem is solved by a continuous descent method. Namely, the gradient is computed analytically by the adjoint state method, then discretized. Assuming Fr\'echet differentiability of the functional, we characterize, analyze and obtain an explicit expression for the gradient. It is perhaps a matter of taste, but we find a neater  proof using the Fr\'{e}chet derivative instead of that of Gateaux.

The paper is organized as follows. In Section \ref{sec:MatMethods}, the shallow water model in channels with varying width, a constrained minimization approach and the derivation of the analytic gradient procedure are presented. Direct and adjoint equations are solved numerically with a second order accurate Roe-type upwind scheme, and the details are presented in  Section \ref{sec:NumMethods} together with the line search method. Section \ref{sec:TestProblems} is devoted to benchmark problems, and the numerical performance and efficiency of the method are studied. A numerical sensitivity analysis starts with a case where the observed velocity corresponds to a steady state (in the absence of friction) and the exact bottom bathymetry consists of a bump plus a sinusoidal perturbation. This simple setting is used to verify the reliability of the algorithm before transient flows with shock waves are considered in further tests. Finally, we invert both the bathymetry and the Manning's friction coefficient simultaneously in a transient flow and we also present a numerical test involving wet-dry states motivated by experimental data. In all cases the approximated bathymetry is very accurate. Conclusions are presented in the last section. Details of the derivation of the gradient is left to Appendix \ref{sec:AppendixGradient}.


\section{Materials and methods}
\label{sec:MatMethods}

\subsection{The shallow water model}

\begin{figure}[h!]
\begin{center}
{\includegraphics[width=0.49\textwidth]{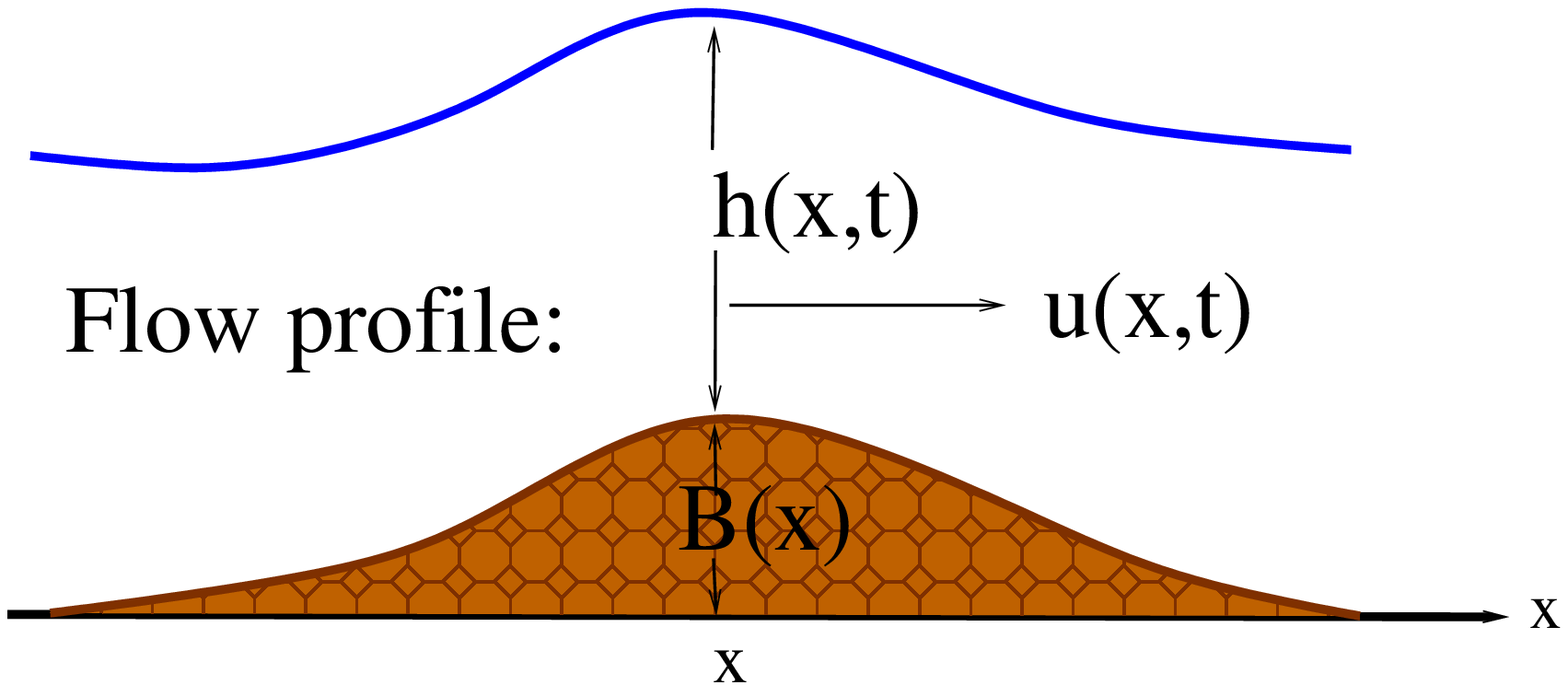}}
{\includegraphics[width=0.4\textwidth]{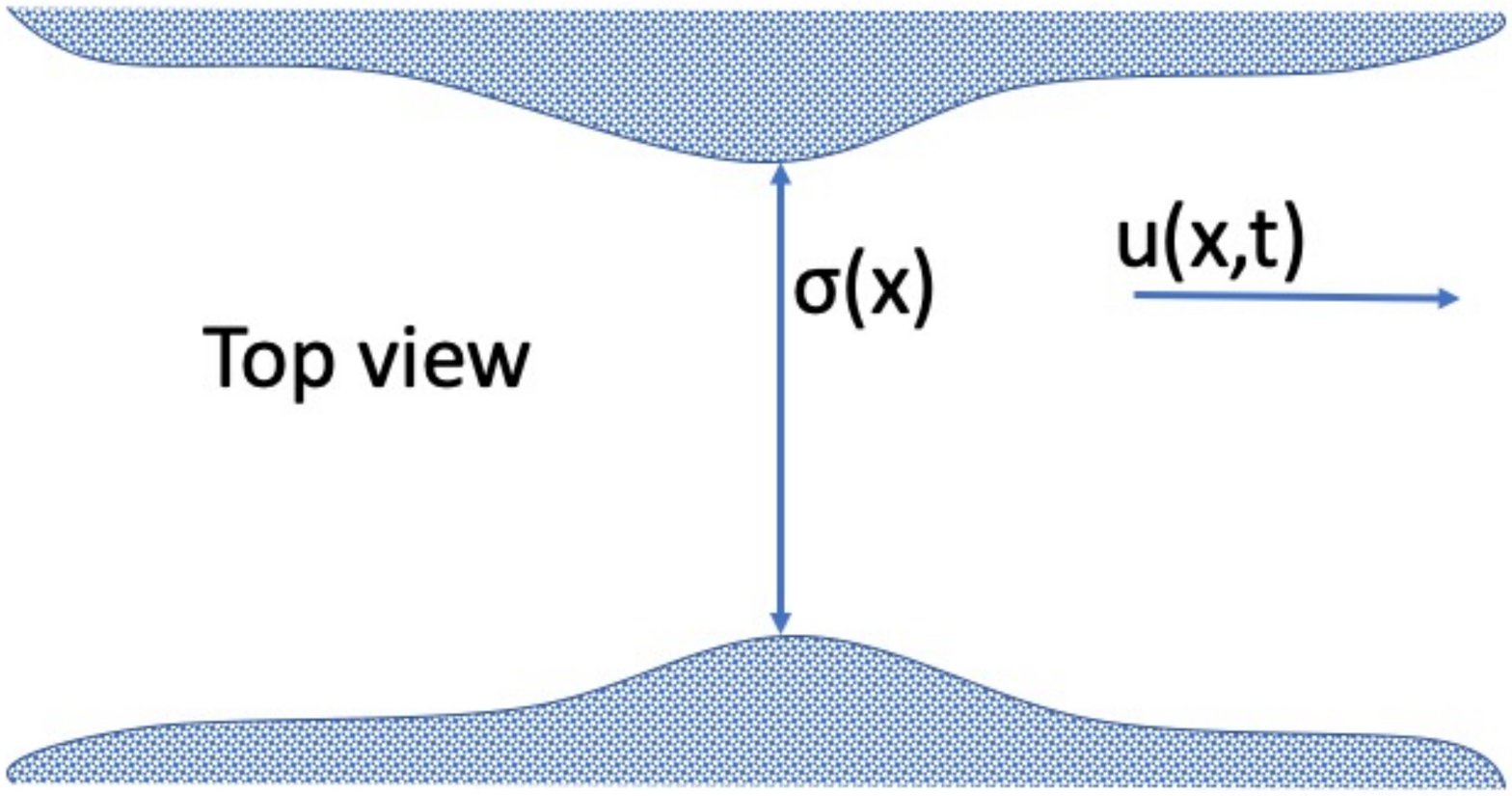}}
\end{center}
\caption{\label{fig:SWSchematic} Schematic of the shallow water model. The left panel shows the flow profile along the channel. The right panel shows the top view.}
\end{figure} 


The direct problem consist of the 1-D SWE for rectangular channels with varying width and friction, which is written as a hyperbolic balance law as
\begin{equation}
\label{eq:SWSigma}
\begin{pmatrix}
\sigma h \\
\sigma hu
\end{pmatrix}_t + 
\begin{pmatrix}
\sigma hu \\
\sigma hu^2 +\frac{g}{2}\sigma h^2 
\end{pmatrix}_x = 
\begin{pmatrix}
0 \\
\frac{1}{2}gh^2\sigma_x-g\sigma h B_x -gn^2\frac{\sigma h}{R^{4/3}}u\vert u\vert
\end{pmatrix}, \; \;  a < x < b,  \; \;  t > 0.
\end{equation}

Here $h$ is the depth of the layer, $u$ the velocity, $B(x)$ the bathymetry, $\sigma(x)$ the channel's width at position $x$, $g=9.81 \text{ ms}^{-2}$ the acceleration of gravity, $n$ is the Manning's friction coefficient, and $R= \frac{ \sigma h}{ \sigma + 2 h}$ the hydraulic radius. The hydraulic radius is the ratio of the wet area and the wetted perimeter. See \cite{khan2014modeling} for more details on the friction term.  Figure \ref{fig:SWSchematic} shows the schematic of the model. The velocity is in units of meters per second, while $x,h,\sigma$ are in units of meters and time in seconds. The friction coefficient is in units of $\text{ s}\text{ m}^{-1/3}$.

The adjoint problem is solved backwards in time. In order to have a well posed problem, appropriate initial and boundary conditions are imposed. The details are included in Section \ref{sec:TestProblems}.

\subsection{A constrained minimization approach and the analytic gradient}

There exists a variety of techniques to measure the velocity in open channels. See for instance \cite{bolognesi2017measurement} where river discharge estimations and measurements of velocity using an aircraft system are analyzed. The inverse problem to solve is stated as follows. 

For simplicity, let us assume that the cross-sectional velocity $u$ is measured at the points $(x_j,t_k),$ $j=1,2,\ldots N,$  $k=1,2,\ldots K$. Namely, $u(x_j,t_k)\approx \hat{u}_{j,k}$. The inverse problem of interest consists in estimating the bathymetry $B\equiv B(x)$, and the Manning's friction coefficient $n$, given this velocity data. For that end, let us introduce the least square functional, 
\begin{equation}
\begin{array}{rcl}
J:L^2(a,b)\times \mathbb{R} & \rightarrow & \mathbb{R} \\
(B,n) & \mapsto & J(B,n),
\end{array}
\end{equation}
given by
\begin{equation}
J(B,n) = \frac{1}{2} \sum_{j,k} \left( u(x_j,t_k,B)-\hat u_{j,k} \right)^2.
\end{equation}
Our goal is to minimize $J$ constrained to $h,u$ solving the shallow water system \eqref{eq:SWSigma}.


The constrained optimization problem is solved by a continuous descent method. Namely, the gradient is computed analytically by the adjoint state method, then discretized. Let us define the (linear) observation operator
\begin{equation}
\label{eq:ObsOp}
\mathcal{M}u=\lbrace u(x_j,t_k) \rbrace\in\mathbb{R}^{J\times K}.
\end{equation}
Let us also consider the Lagrangian
\begin{equation}
\begin{array}{l}
\mathcal L(B,n,h,u,\lambda,\mu)   =   \frac{1}{2}\Vert \mathcal{M}u-\hat{u}\Vert^2 + \\
   \\
 \left\langle 
\left( 
\begin{array}{c}
\lambda \\ 
\mu
\end{array}
\right),
\left( 
\begin{array}{c}
(\sigma h)_t+(\sigma hu)_x \\
(\sigma hu)_t+(\sigma hu^2+\frac{g}{2}\sigma h^2)_x-\frac{g}{2}h^2\sigma_x+g\sigma h B_x +gn^2\frac{\sigma h}{R^{4/3}}u\vert u\vert
\end{array}
\right)
\right\rangle_{L^2(\Omega\times (0,T))}, 
\end{array}
\end{equation}
where $\Omega=(a,b)$, and $\lambda,\mu$ are Lagrange multipliers. Here $\langle \cdot, \cdot \rangle_{L^2(\Omega\times (0,T))}$ is the normalized inner product in $L^2(\Omega\times (0,T))$, given by
\[
\langle f, g \rangle_{L^2(\Omega\times (0,T))} = \frac{1}{(b-a) T} \int_0^T \int_a^b f(x,t) \overline{g(x,t)} dx dt.
\]

\bigskip

Assuming Fr\'{e}chet differentiability of $J$, the next result yields an  expression for the gradient.

\begin{theorem}
\label{th:Lagrange}
Let $B$ and $n$ be given, and let $h,u$ solve the shallow water equations. Suppose that the velocity measurements are taken in space and time $(x,t)\in [a,b]\times [0,T]$ for $T>0.$ Furthermore, assume the Lagrange multipliers solve the adjoint equations
\begin{equation}
\label{eq:AdjointEqn}
\begin{array}{c}
\sigma \lambda_t+\sigma u\lambda_x
+\sigma u\mu_t+(\sigma u^2 +g\sigma h)\mu_x 
+\left(gh\sigma_x-g\sigma B_x  -
gn^2\left(\frac{1}{h}+\frac{2}{\sigma}\right)^{1/3}\left(2-\frac{1}{3}\frac{\sigma}{h}\right) u\sqrt{u^2+\varepsilon}\right) \mu
 =  0 \\
\sigma h \mu_t+2\sigma hu\mu_x+\sigma h \lambda_x
-gn^2\frac{\sigma h}{R^{4/3}}\frac{2u^2+\varepsilon}{\sqrt{u^2+\varepsilon}}\mu
 =  \mathcal{M}^*(\mathcal{M}u-\hat{u}) ,
\end{array}
\end{equation}
with final and boundary conditions
\begin{equation}
\label{eq:LambdaInitCond}
\lambda(x,T)=0,\quad \mu(x,T)=0,\quad x\in(a,b),
\end{equation}
and
\begin{equation}
\label{eq:MuInitCond}
\lambda(b,t)=0,\quad \mu(b,t)=0.\quad t\in(0,T).
\end{equation}
Then the Fr\'{e}chet derivative of the functional $J$ is
\[
DJ(B,n)(\xi_1,\xi_2)= \langle\xi_1,-\int_0^T(g\sigma h\mu)_x\, dt\rangle\,
+ \langle\mu,2g \; n  \frac{\sigma h}{R^{4/3}} u\sqrt{u^2+\varepsilon}\rangle\xi_2.
\]
Consequently
\[
\nabla J(B,n) = \left(
\begin{array}{c}
-\overline{(g \sigma h \mu )}_x \\ 
\langle\mu,2g \; n\frac{\sigma h}{R^{4/3}}u\sqrt{u^2+\varepsilon}\rangle
\end{array}
\right)
\]
\end{theorem}

\bigskip

Here, $\overline{(\cdot)}$ denotes the time average
\[
 \bar f(x) = \frac{1}{T} \int_0^T f(x,t) dt. 
\]

The proof of this theorem is left to Appendix \ref{sec:AppendixGradient}. 

We note that if the direct problem is solved for times $t\in [0,T]$ with initial conditions at $t=0$ as specified in Section \ref{sec:TestProblems}, the adjoint problem has zero final conditions at time $t=T$ and the solution is found backwards in time from $t=T$ to $t=0$.

\section{Numerical Methods}
\label{sec:NumMethods}

\subsection{A Roe-type scheme for the hyperbolic systems \eqref{eq:SWSigma} and \eqref{eq:AdjointEqn}}

A wide variety of numerical schemes have been proposed to solve the shallow water equations. Such schemes use different approaches and satisfy desirable properties for different goals. In \cite{garcia2000numerical}, a Roe-type well-balanced numerical scheme is proposed. That is, it exactly preserves steady sates at rest, adding accuracy when computing near steady-state flows. The central-upwind scheme presented in \cite{kurganov2007second} satisfies both the well-balance and the positivity-preserving properties. That is, it recognizes steady states at rest and it maintains the positivity of layer's depth over time. Many other approaches have also been studied. We refer the interested reader to the above works and references therein for more information. \\

\noindent
{\bf The direct problem}

In quasilinear form, it is well known (see e.g. \cite{balbas2009central}) that system \eqref{eq:SWSigma} can be written as 
\begin{equation}
\label{eq:MatrixA}
\vect W_t + 
A
\vect W_x = 
\vect S
\end{equation}
where
\begin{equation}
\label{eq:Variables}
\vect W = 
\begin{pmatrix}
\sigma h \\
\sigma hu
\end{pmatrix}, 
A = 
\begin{pmatrix}
0 & 1 \\
c^2-u^2 & 2 u
\end{pmatrix} ,
\text{ and }
\vect S =
\begin{pmatrix}
0 \\
-g \sigma h B_x + g h^2 \sigma_x -\frac{g n^2 A}{R^{4/3}} |u| u
\end{pmatrix}
\end{equation}
are the vector of conserved variables, the coefficient matrix, and the vector of source terms respectively. The coefficient matrix has eigenvalues $\lambda_1 = u-c, \lambda_2 = u+c$, and corresponding eigenvectors $r_1 = (1, u-c )^T$ and $r_2 = (1,u+c)^T$, where $c=\sqrt{gh}$. As a result, system \eqref{eq:SWSigma} is conditionally hyperbolic provided $h > 0$. Hyperbolicity is lost when $h = 0$ in a dry state. 

Roe-type upwind schemes were first introduced in \cite{roe1981approximate}. The numerical scheme requires the computation of a Roe matrix $\bar A (\vect W_\ell, \vect W_r)$ for any left and right states $\vect W_\ell$ and $\vect W_r$. The flux $\vect F = \vect F(W,\sigma) = \left( \sigma hu , \sigma hu^2 +\frac{g}{2}\sigma h^2 \right)^T$ of the model in conservation form \eqref{eq:SWSigma} depend explicitly on the solution variables but also on the model parameter $\sigma$. For such flux, the Roe matrix $\bar A(\vect W_\ell, \vect W_r)$ must satisfy $\bar A(\vect W_\ell, \vect W_r) \to A(\vect W)$ as $\vect W_\ell, \vect W_r \to \vect W$, it must have real eigenvalues with a complete set of eigenvectors, and
\begin{equation}
\Delta \vect F = \bar A(\vect W_\ell,\vect W_r) \Delta \vect W + \left( 0 , -g \bar{h}^2 \Delta \sigma/2 \right)^T,
\end{equation}
where $\Delta \vect F = \vect F(\vect W_r)-\vect F (\vect W_\ell)$, $\Delta \vect W = \vect W_r  - \vect W_\ell$, $\Delta \sigma = \sigma_r - \sigma_\ell$, and $\bar{h}$ is an approximation of $h$ between the left and right states. One such matrix is given by the following Roe linearizations
\[
\bar u = \frac{\sqrt{\sigma_\ell h_\ell} \; u_\ell + \sqrt{\sigma_r h_r} \; u_r}{\sqrt{\sigma_\ell h_\ell}+\sqrt{\sigma_r h_r}},\; \bar h = \frac{\sqrt{\sigma_\ell}\; h_\ell + \sqrt{\sigma_r} \; h_r }{\sqrt{\sigma_\ell} + \sqrt{\sigma_r} }  \text{, and } \bar c = \sqrt{g \bar h}.
\]

In \cite{garcia2000numerical}, a Roe-type upwind scheme is derived with the aid of a convenient discretization of the source terms that balance the flux gradients for steady states at rest. That is, the numerical scheme is well balanced. See \cite{hubbard2000flux} for more details. In order to extend it here for the case of channels with varying width, one possible discretization of the source terms is given by
\[
\Delta x \;  \overline{\vect S} = 
\begin{pmatrix}
0 \\
-g \bar{\sigma} \bar h \Delta B + g \bar h^2 \Delta \sigma - \frac{g n^2 \bar{\sigma} \bar{h}}{\bar{R}^{4/3}} |\bar u | \bar u_,
\end{pmatrix},
\]
where
\[
\bar \sigma = \sqrt{\sigma_\ell \sigma_r}, \text{ and } \bar R= \frac{\bar \sigma \bar h}{\bar \sigma + 2\bar h}, \Delta B = B_r-B_\ell, \Delta \sigma = \sigma_r-\sigma_\ell.
\]

Finite differences of the conserved variables and the linearized source terms are decomposed in terms of the eigenvectors $\bar{ \vect  r}_{1} =  (1,\bar u - \bar c )^T$, $\bar{ \vect  r}_{2} =  (1,\bar u + \bar c )^T$ as
\[
\Delta \vect W  = \alpha_{1} \bar { \vect  r}_{1 }+ \alpha_{2} \bar{\vect  r}_{2}, \; \; 
\Delta x \widehat{\vect S} = \beta_{1} \bar{ \vect r}_{2}+ \beta_{1} \bar{ \vect r}_{2},
\]
where
\[
\begin{array}{lcllcl}
\alpha_{1} & = &  \frac{- \Delta (\sigma h u) + (\bar u + \bar c)  \Delta (\sigma h )}{2 \bar c}, & \beta_{1} &  = & \frac{\bar c \bar \sigma}{2} \Delta B -\frac{ \bar c \bar h}{2} \Delta \sigma + \frac{n^2 \bar c \bar \sigma}{2 \bar R^{4/3}} \; \Delta x \; |\bar u| \bar u, \\
\alpha_{2} & = &  \frac{\Delta (\sigma h u) - (\bar u - \bar c)  \Delta (\sigma h)}{2 \bar c}, & \beta_{2} &  = & -\frac{\bar c \bar \sigma}{2} \Delta B + \frac{ \bar c \bar h}{2} \Delta \sigma - \frac{n^2 \bar c \bar \sigma}{2 \bar R^{4/3}} \; \Delta x \; |\bar u| \bar u.
\end{array}
\]

More details on the implementation of the numerical scheme can be found in \cite{roe1987upwind}. For the sake of completeness, we include the details here. We denote by $\vect W_{j+\half}$ the Roe averages between the states $\vect W_j$ and $\vect W_{j+1}$ in the domain with cells $I_{j}=[x_{j-\half},x_{j+\half}]$, $x_{j\pm \half}= x_j \pm \Delta x/2$. The second order numerical scheme is given by 
\begin{equation}
\label{eq:Upwind}
\begin{array}{lcl}
\vect W_j^{k+1} & = &  \vect W_j^k - \frac{\Delta t}{\Delta x} \sum_{\lambda_{j-1/2,p}^k > 0} (\alpha_{j-1/2,p}^k \lambda_{j-1/2,p}^k-\beta_{j-1/2,p}^k )\; \bar{ \vect r}_{j-1/2,p}^k \\
 & & -\frac{\Delta t}{\Delta x} \sum_{\lambda_{j+1/2,p}^k < 0} ( \alpha_{j+1/2,p}^k \lambda_{j+1/2,p}^k -\beta_{j+1/2,p}^k ) \; \bar{ \vect r}_{j+1/2,p}^k \\
& & -\sum_{p=1}^2 \frac{\Delta t}{\Delta x} \frac{1}{2} \phi \left( \theta_{j+1/2,p}^k \right) \left( \text{sign}(\nu_{j,p}^k) - \nu_{j,p}^k \right) \left( \alpha_{j+1/2,p}^k \lambda_{j+1/2,p}^k-\beta_{j+1/2,p}^k \right)\;  \bar{ \vect r}_{j+1/2,p}^k \\
 & & +\sum_{p=1}^2 \frac{\Delta t}{\Delta x} \frac{1}{2} \phi \left( \theta_{j-1/2,p}^k \right) \left( \text{sign}(\nu_{j-1,p}^k) - \nu_{j-1,p}^k \right) \left( \alpha_{j-1/2,p}^k \lambda_{j-1/2,p}^k-\beta_{j-1/2,p}^k \right)\;  \bar{\vect r}_{j-1/2,p}^k .
 \end{array}
\end{equation}
Here, $\phi(\theta) = \max( 0, \max (\min (2 \theta,1) , \min( \theta , 2 ) ) )$ is known as the superbee limiter function, and for each cell $I_j$, we define
\[
\theta_{j+1/2,p} = \frac{\alpha_{j,p}\lambda_{j,p}-\beta_{j,p}}{\alpha_{j',p}\lambda_{j',p}-\beta_{j',p}}, \; j' = j-\text{sign}(\lambda_{j,p}), \; \nu_{j,p} = \frac{\Delta t}{\Delta x}\lambda_{j,p}.
\]
The last two terms in equation \eqref{eq:Upwind} are the second order corrections. See \cite{leveque1992numerical} for more details, where the reader can also find the sonic entropy fix that is usually done for Roe-type upwind schemes and that has also been implemented here. In the case where $\lambda_{j-1,p} < 0 < \lambda_{j,p}$, $\lambda_{j-1/2,p}$ in the first term of equation \eqref{eq:Upwind} is replaced by $\lambda_{j-1/2,p}^r = \lambda_{j,p} \; (\lambda_{j-1/2,p}-\lambda_{j-1,p})/(\lambda_{j,p}-\lambda_{j-1,p})$. Symmetrically, if $\lambda_{j,p} < 0 < \lambda_{j+1,p}$, $\lambda_{j+1/2,p}$ in the second term of equation \eqref{eq:Upwind}, $\lambda_{j+1/2,p}$ is replaced by $\lambda_{j+1/2,p}^\ell = \lambda_{j,p} \; (\lambda_{j+1,p}-\lambda_{j+1/2,p})/(\lambda_{j+1,p}-\lambda_{j,p})$.\\

\noindent
{\bf The adjoint problem}

The adjoint problem can be written as
\begin{dmath}
\begin{pmatrix}
\lambda \\
\mu
\end{pmatrix}_t
+
\begin{pmatrix}
0 & c^2 - u^2 \\
1 & 2 u
\end{pmatrix}
\begin{pmatrix}
\lambda \\
\mu
\end{pmatrix}_x
=
\begin{pmatrix}
-\frac{u}{\sigma h} \mathcal M^* (\mathcal M u - \hat u) + \left[ -\frac{g h}{\sigma} \sigma_x + g B_x  +g n^2 u \left( \frac{ 2 h -\sigma/3 }{\sigma h \; R^{1/3}} \sqrt{u^2+\epsilon} - \frac{2 u^2+\epsilon}{R^{4/3} \sqrt{u^2+\epsilon}} \right)\right] \mu \\
\frac{1}{\sigma h} \mathcal M^* (\mathcal M u - \hat u) + \frac{g n^2}{ R^{4/3}} \frac{2u^2+\epsilon}{\sqrt{u^2 + \epsilon}} \mu
\end{pmatrix},
\end{dmath}
where $\hat u $ is the observed velocity in space and time. We note that $\lambda$ has units of velocity and $\mu$ is non-dimensional. On the other hand, $M^* (\mathcal M u - \hat u)$ is given in units of squared velocity. 

As noted in Theorem \ref{th:Lagrange}, the final conditions in equations \eqref{eq:LambdaInitCond} and \eqref{eq:MuInitCond}, are given at time $t=T$ and the solution is computed backwards in time from $t=T$ to $t=0$. In addition, we note that the coefficient matrix
\[
A^* =
\begin{pmatrix}
0 & c^2 - u^2 \\
1 & 2 u
\end{pmatrix}
\]
is the transpose of the coefficient matrix in the direct problem. The  eigenvalues are the same and the corresponding eigenvectors are
\[
\vect r_1^* = 
\begin{pmatrix}
-\bar u - \bar c \\
1 
\end{pmatrix},
\; \text{ and } \;
\vect r_2^* = 
\begin{pmatrix}
-\bar u + \bar c  \\
1 
\end{pmatrix}.
\]

Analogously to the direct problem, the finite difference of the solution variable $\Delta \vect W^* = (\Delta \lambda, \Delta \mu)^T$ and the linearized source terms 
\begin{dmath}
\Delta x \; \vect S^*=
\begin{pmatrix}
-\frac{\bar u \; \Delta x \; \mathcal \overline{ M^* (\mathcal M u - \hat u}) }{\bar \sigma \bar h } + \left[ -\frac{g \bar h}{\bar \sigma} \Delta \sigma + g \Delta B  +g n^2 \bar u \left( \frac{ 2 \bar h -\bar \sigma/3 }{\bar \sigma \bar h \; \bar R^{1/3}} \sqrt{\bar u^2+\epsilon} - \frac{2 \bar u^2+\epsilon}{\bar R^{4/3} \sqrt{\bar u^2+\epsilon}} \right)\right] \bar \mu \\
\frac{1}{\bar \sigma \bar h} \overline{ \mathcal M^* (\mathcal M u - \hat u}) + \frac{g n^2}{ \bar R^{4/3}} \frac{2 \bar u^2+\epsilon}{\sqrt{\bar u^2 + \epsilon}} \bar \mu
 \end{pmatrix},
 \end{dmath}
 where
 \[
 \bar \mu = \frac{\mu_\ell + \mu_r}{2}, ; \; \overline{  \mathcal M^* ( \mathcal M u - \hat u )} = \frac{  \mathcal M^*( \mathcal M u_\ell - \hat u_\ell) +   \mathcal M^*( \mathcal M u_r - \hat u_r )}{2}
 \]
 are decomposed as
 \[
\Delta \vect W^* = \alpha_{1}^* \bar { \vect  r}_{1}^*+ \alpha_{2}^* \bar{\vect  r}_{2}^*, \; \; 
\Delta x \widehat{\vect S}^* = \beta_{1}^* \bar{ \vect r}_{2}^* + \beta_{1}^* \bar{ \vect r}_{2}^*.
\]
Here, the coefficients in the decompositions are given by
\[
\begin{array}{lcllcl}
\alpha_1^* & = &  \frac{(-\bar u + \bar c)  \Delta \mu - \Delta \lambda}{2 \bar c}, & \beta_1^* &  = &  \frac{(-\bar u + \bar c)  \bar S_2^* - \bar S_1^*}{2 \bar c}, \\ \\
\alpha_2^* & = &  \frac{(\bar u + \bar c)  \Delta \mu + \Delta \lambda}{2 \bar c}, & \beta_2^* &  = & \frac{(\bar u + \bar c)  \bar S_2^* + \bar S_1^*}{2 \bar c}.
\end{array}
\]

The corresponding numerical scheme for the adjoint problem solved backward in time is given by
\begin{equation}
\label{eq:UpwindAdjoint}
\begin{array}{lcl}
\vect W_j^{*,k} & = &  \vect W_j^{*,k+1} + \frac{\Delta t}{\Delta x} \sum_{\lambda_{j-1/2,p}^{k+1} > 0} (\alpha_{j-1/2,p}^{*,k+1}  \lambda_{j-1/2,p}^{k+1}-\beta_{j-1/2,p}^{*,{k+1}} )\;  \vect r_{j-1/2,p}^{*,k+1} \\
 & & +\frac{\Delta t}{\Delta x} \sum_{\lambda_{j+1/2,p}^{k+1} < 0} ( \alpha_{j+1/2,p}^{*,k+1} \lambda_{j+1/2,p}^{k+1} -\beta_{j+1/2,p}^{*,k+1} ) \; \vect r_{j+1/2,p}^{*,k+1} \\
& & -\sum_{p=1}^2 \frac{\Delta t}{\Delta x} \frac{1}{2} \phi \left( \theta_{j+1/2,p}^{*,k+1} \right) \left( \text{sign}(\nu_{j,p}^{k+1}) - \nu_{j,p}^{k+1} \right) \left( \alpha_{j+1/2,p}^{*,k+1} \lambda_{j+1/2,p}^{k+1}-\beta_{j+1/2,p}^{+,k+1} \right)\;  \vect r_{j+1/2,p}^{*,k+1} \\
 & & +\sum_{p=1}^2 \frac{\Delta t}{\Delta x} \frac{1}{2} \phi \left( \theta_{j-1/2,p}^{*,k+1} \right) \left( \text{sign}(\nu_{j-1,p}^{k+1}) - \nu_{j-1,p}^{k+1} \right) \left( \alpha_{j-1/2,p}^{*,k+1} \lambda_{j-1/2,p}^{k+1}-\beta_{j-1/2,p}^{*,k+1} \right)\;  \vect r_{j-1/2,p}^{*,k+1}, 
 \end{array}
\end{equation}
where
\[
\theta_{j+1/2,p}^* = \frac{\alpha_{j,p}^*\lambda_{j,p}-\beta_{j,p}^*}{\alpha_{j',p}^* \lambda_{j',p}-\beta_{j',p}^*}.
\]


\subsection{A line search method}
\label{sec:SearchMethod}

The bathymetry and the Manning's friction coefficient are inferred iteratively. We start with an initial guess which in principle must be not too far from the target. In each step, one computes the gradient and advance in the steepest search direction. The amplitude to advance in the steepest direction is initially obtained empirically. One then modulates it to minimize the error. We continue iteratively until an error threshold is achieved. The algorithm is summarized as follows.\\
 
\noindent \textbf{Algorithm.} (Continuous descent)

\label{alg:descent}
Given a starting point $B_0$ and $n_o$, a convergence
tolerance $\epsilon$, and $k\leftarrow 0;$

\noindent 
while $\left\Vert \nabla J(B _{k},n_k)\right\Vert >\epsilon ;$

Compute the steepest search direction
\begin{equation}
\label{eq:pk}
p_{k} =  -\nabla J(B_k,n_k);
\end{equation}
Set 
\begin{equation}
\label{eq:SteepestDescent}
B _{k+1} = B _{k}+\alpha_k p_{1,k}; \; \; \;
n_{k+1}  = n_k+\alpha_k p_{2,k};
\end{equation}
$k\leftarrow k+1;$
\noindent 
end (while)\\[4pt]

Since
\[
\nabla J(B,n) = \left(
\begin{array}{c}
-\overline{(g \sigma h \mu )}_x \\ 
\langle\mu,2g\frac{\sigma h}{R^{4/3}}u\sqrt{u^2+\varepsilon}\rangle
\end{array}
\right),
\]
we have
\[
p_k= 
\begin{pmatrix}
\overline{(g \sigma h \mu )}_x \\ 
-\langle\mu,2g\frac{\sigma h}{R^{4/3}}u\sqrt{u^2+\varepsilon}\rangle
\end{pmatrix}.
\]

We recall that the overline denotes time average. The steepest search direction for the bathymetry is a function of $x$ only, and a constant for the Manning's friction coefficient. At each iteration step where $p_k$ from equation \eqref{eq:pk} is already computed, we choose the coefficient $\alpha_k$ in equation \eqref{eq:SteepestDescent} as follows. One starts with an initial value ($\alpha_k = 0.5$ here) and compute $(B_{k+1},n_{k+1})$ according to equation \eqref{eq:SteepestDescent}. One then calculates the error in the velocity with that estimated bathymetry. If the error decreases when when $\alpha_k$ is reduced by a certain factor ($0.8$ here), we keep reducing it until the error does not decrease anymore. \\

\noindent
{\bf Note:} When either the Manning's friction coefficient or the bathymetry are known, we can estimate the other parameter by considering only one of the entries in $p_k$.

\section{Test problems}
\label{sec:TestProblems}

The above technique is numerically tested in this section. The direct problem often involves bathymetries consisting of a bump or a channel's throat by which the fluid passes through. Depending on the parameter regime, the flow may accelerate/decelerate and reduce/increase its cross-sectional wet area as it passes through the bump and/or throat. We consider here different situations to show the merits and robustness of the algorithm.\\

\noindent
{\bf Numerical setup and boundary conditions}\\

The inverse problem consists of inferring the bathymetry $B$ and Manning's friction coefficient $n$ from the transient velocity $u(x,t)$. We assume that the velocity is observed at all spatial positions and at all times. That is, the velocity data is assumed to be available at all grid points and at all times. For the direct problem, we initially specify the total height $w$ and the velocity $u$. Such data is assumed to be known at $t=0$. At each step, the estimated bathymetry is used to obtain the initial depth $h=w-B$. Regarding the adjoint problem which is solved backwards in time, here we impose zero Dirichlet final conditions. 

At the left boundary, a discharge $Q_{\text{left}}$ and a surface elevation $w_{\text{left}}$ are specified at inflow and are extrapolated at outflow for the direct problem. An inflow/outflow at the left boundary occurs when the eigenvalues of the coefficient matrix are positive/negative. At the right boundary, a discharge $Q_{\text{right}}$ and a surface elevation $w_{\text{right}}$ are specified at inflow and extrapolated at outflow. An inflow/outflow at the right boundary occurs when the eigenvalues of the coefficient matrix are negative/positive. The adjoint problem is used to compute the gradient $\nabla J$. Zero Neumann boundary conditions are implemented for the adjoint variables $\lambda,\mu$. We assume we know the bathymetry elevation at the boundaries, and we prescribe them to be $B_{\text{in}}=B_{\text{out}}=0$ at both ends. 

We quantify the error and the relative error with the $L^\infty$ norm, and are given by
\[
e = \sup_x |B_{\text{approx}}(x) - B_{\text{exact}} | , \text{ and } e_{\text{rel}} = \frac{e}{\sup_x | B_{\text{exact}} (x)| },
\]
where $B_{\text{approx}}$ and $B_{\text{exact}}$ are the approximated and exact bathymetries. 

The time window $[0,T]$ where both the direct and the adjoint problems are solved need to be chosen carefully. The end time $T$ needs to be large enough to have the needed information to invert the problem. However, if $T$ is too large, it induces strong interactions with the boundary, where the bathymetry is prescribed. In any case, we have found that the bathymetry estimation is not very sensitive to the end time. We specify the end time $T$ in each numerical test. 

\subsection{Bathymetry bump}

\begin{center}
\begin{figure}[h]
\centering
\includegraphics[width=0.49 \textwidth]{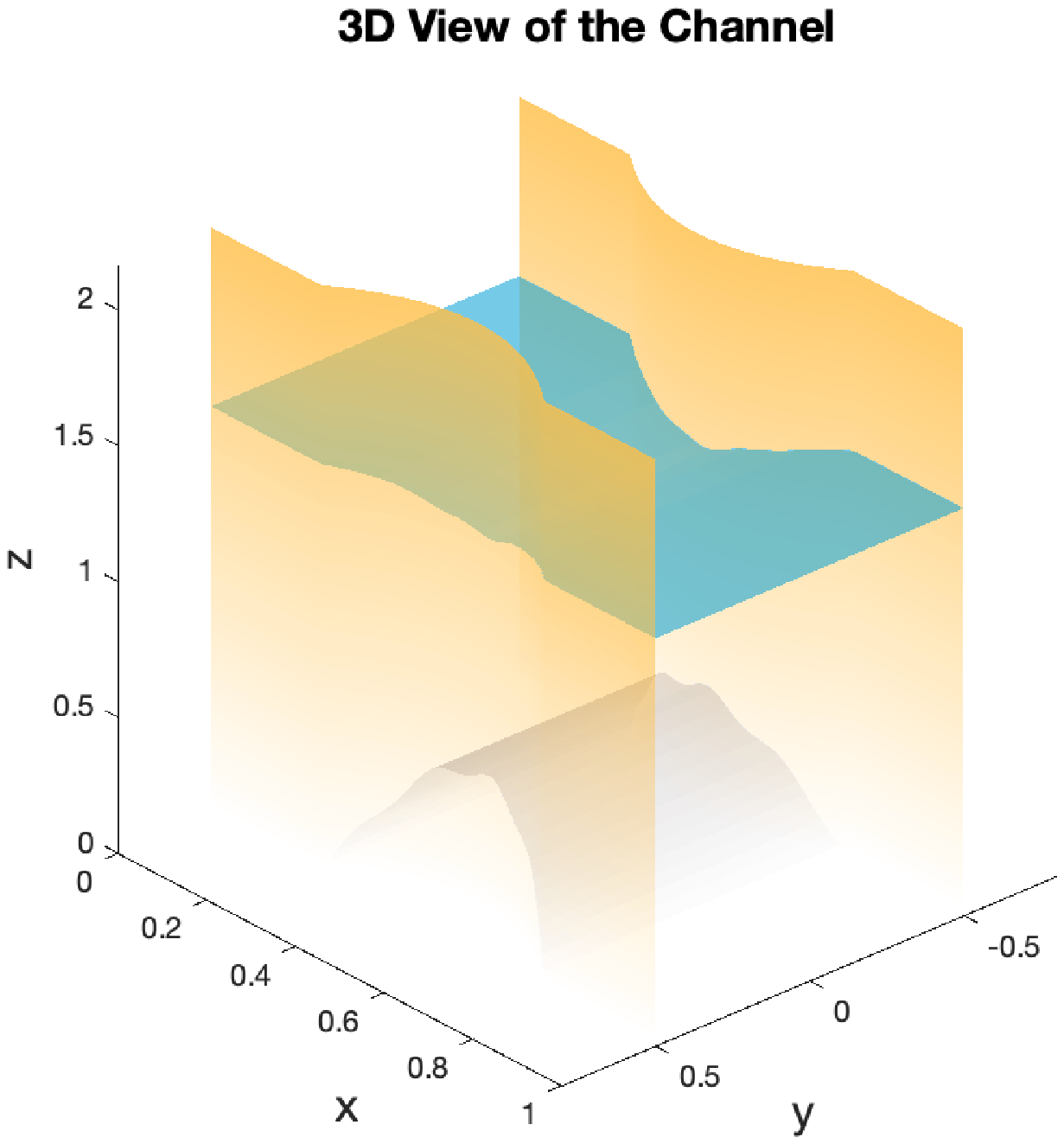}
\includegraphics[width=0.49 \textwidth]{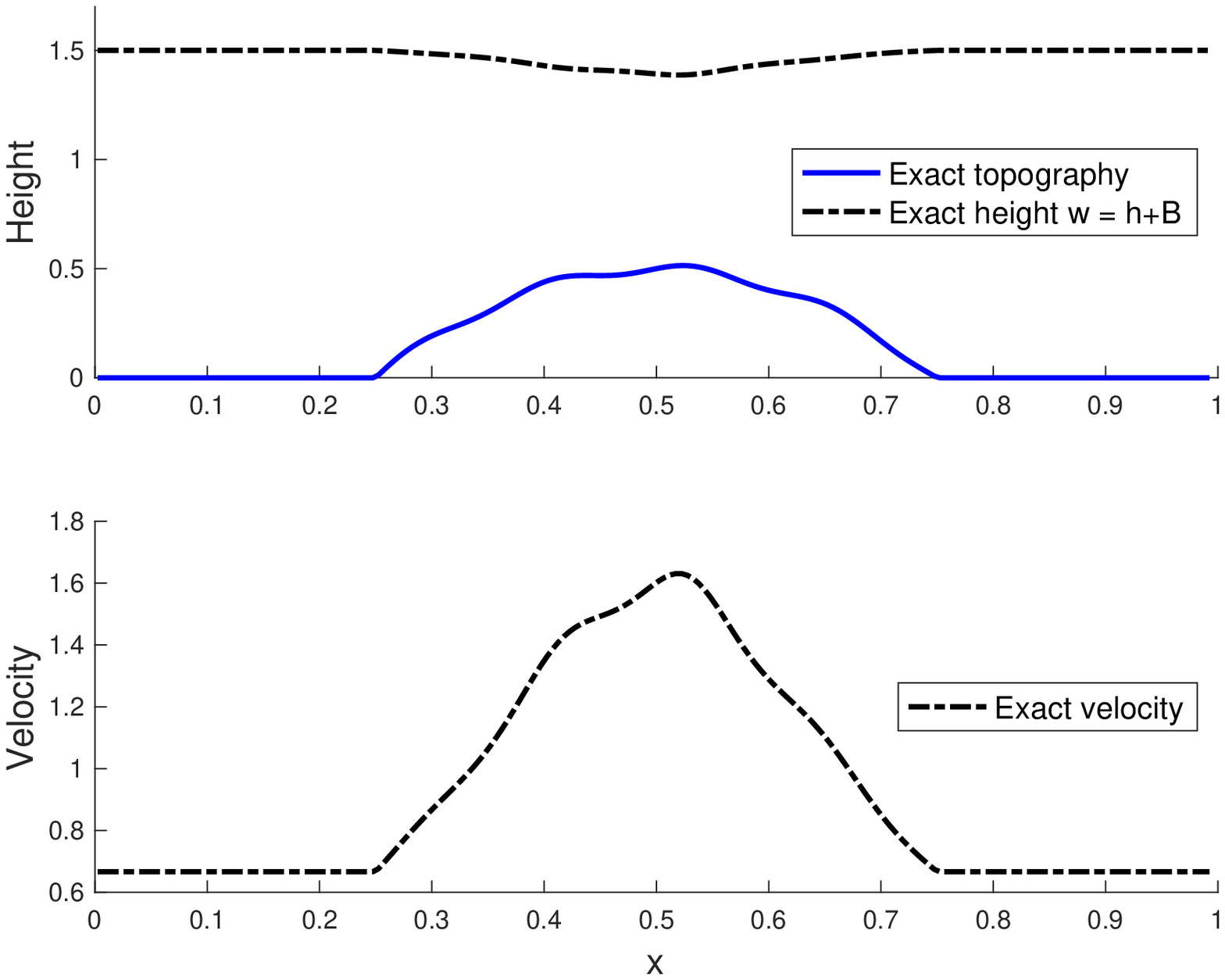}
\caption{\label{fig:BumpSinSubcr} Left panel: 3D view of the channel. Top right panel: Exact total height $w=h+B$ (blue solid line) and bathymetry (black solid line) are shown. Bottom right panel: Velocity as a function of $x$. The solution corresponds to a steady state with discharge $Q_{\text{in}} = 1$ at inflow (left boundary) and total height $w_{\text{out}} = 1.5$ at outflow (right boundary).}
\end{figure}
\end{center}

The synthetic data in this first numerical test is obtained with a particular choice of a bathymetry elevation and channel's geometry. The exact bathymetry to be estimated is given by
\begin{equation}
\label{eq:BumpBExact}
B_{\text{exact}}(x) = 
\left\{
\begin{array}{lcl}
\frac{1}{2}  \left( 1-\left( 4 \left( x-\frac{1}{2} \right) \right)^2 \right) + 0.02 \sin \left( 16 \pi \left( x - \frac{1}{4} \right) \right)  & \text{ if } x \in [0.25 , 0.75]\\
 0 & \text{ if } x \in [0,1] \setminus [0.25, 0.75].
\end{array}
\right.
\end{equation}
The channel's width is given by
\begin{equation}
\label{eq:sigma}
\sigma(x) = 2 \min \left( 2.4 ( x - 0.5)^2+0.35 , 0.5 \right). 
\end{equation}
The 3D view of the channel is shown in the left panel of Figure \ref{fig:BumpSinSubcr}. Following \cite{khan2014modeling}, the Manning's friction coefficient is fixed to $n=0.009 \text{ s}\text{ m}^{-1/3}$ in this numerical test and the bathymetry is the only model's parameter to estimate. 

We first test the algorithm in a simple setting. In particular, the velocity considered here corresponds to a subcritical smooth steady state (in the absence of friction). Smooth steady states are characterized by two invariants. Namely, the discharge $Q = hu$ and the energy $E=\frac{1}{2} u^2 + g(h+B)$ are both constant throughout the domain when the friction coefficient $n$ vanishes. One could use such invariants to estimate the bathymetry. However, the algorithm presented here is designed for transient flows as well, and the setting in this numerical experiment is meant to test the accuracy in the approximations of the bathymetry. Given the bathymetry $B$, a corresponding steady state may be computed by specifying two quantities. Here we specify the discharge $Q_{\text{in}}=1$ at inflow at the left boundary and the total height $w_{\text{out}}= B_{\text{out}}+h_{\text{out}}=1.5$ at outflow at the right boundary. Figure \ref{fig:BumpSinSubcr} shows the total height $w=h+B$ and exact bathymetry $B$ in the top right panel, while the bottom right panel shows the exact velocity that is observed. 

\begin{center}
\begin{figure}[h!]
\centering
\includegraphics[width=0.49 \textwidth]{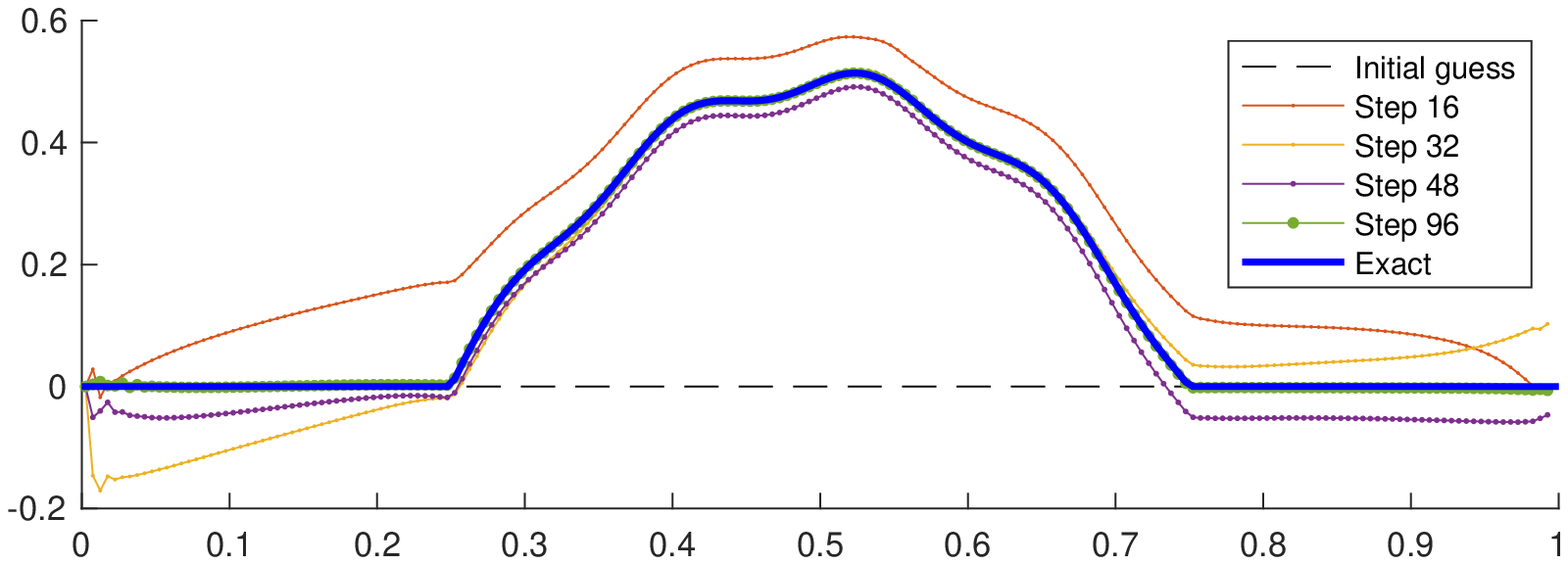}
\includegraphics[width=0.49 \textwidth]{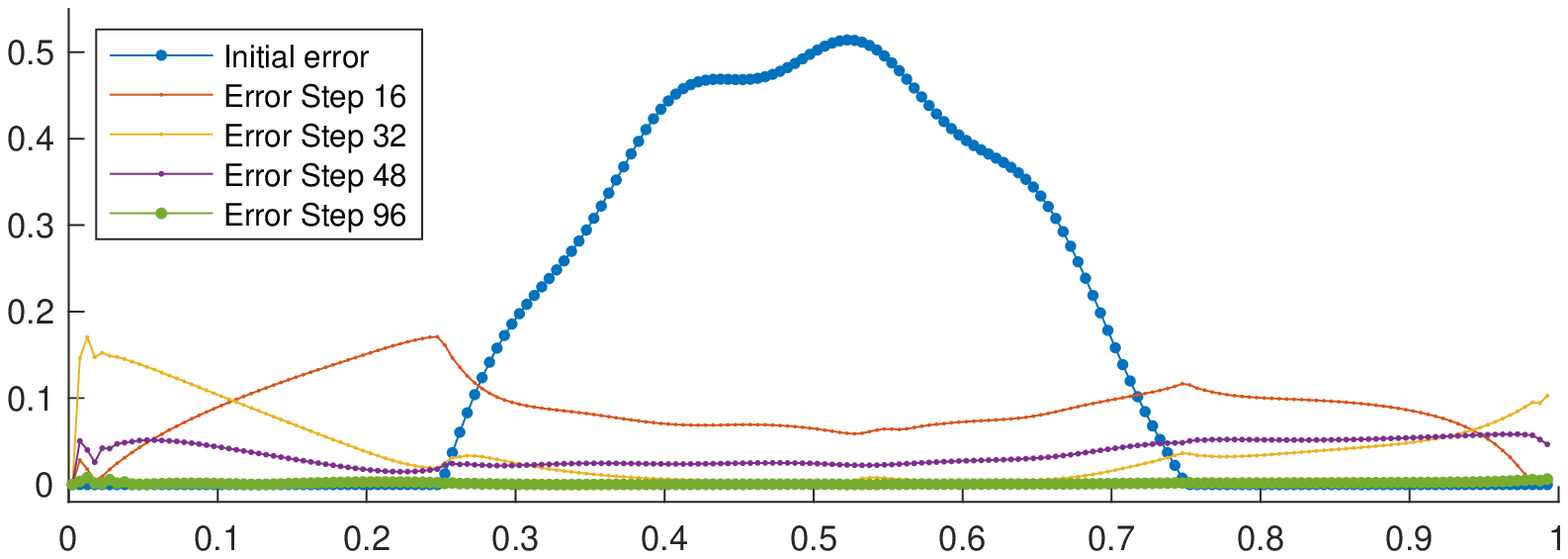}\\
\caption{\label{fig:BumpSolution} The top left panel shows the exact bathymetry (solid blue line) given by equation \eqref{eq:BumpBExact}, the initial guess $B_o = 0$ (black dashed line) and the intermediate steps (dotted lines in different colors and mark sizes). The top right panel shows the error $|B_n - B_{\text{exact}}|$ at steps $16,32,48$ and $96$.}
\end{figure}
\end{center}

In general, the initial guess needs to be close enough to the target in order to converge to the correct state. However, here we test the robustness of the algorithm by considering an initial guess $B_o = 0$. The numerical results for the inverse problem are given in Figure \ref{fig:BumpSolution}. In the left panel, the exact solution is identified with the solid blue line, while the initial guess is denoted by the black dashed line. The final approximation is computed using 96 steps in the algorithm described in Section \ref{sec:SearchMethod} and a resolution of 200 grid points. The direct and adjoint problems are solved in the time window $[0,T]$ with $T=0.2$. In the panel, we show the estimation of the bathymetry for the steps 0,16,32,48, and 96. The final step is not easy to distinguish because it is very close to the target. The error is shown in the right panel. The maximum error is $e=8.1 \times 10^{-3}$, which corresponds to a relative error of $e_{\text{rel}} = 1.56 \%$. The maximum error is located near the left boundary. Away from that region, the error is reduced to $3.1 \times 10^{-3}$, which corresponds to a relative error of $0.59 \%$.

\subsection{A transient flow}
\label{sec:TestTransientFlow}

\begin{center}
\begin{figure}[h]
\centering
\includegraphics[width=0.49 \textwidth]{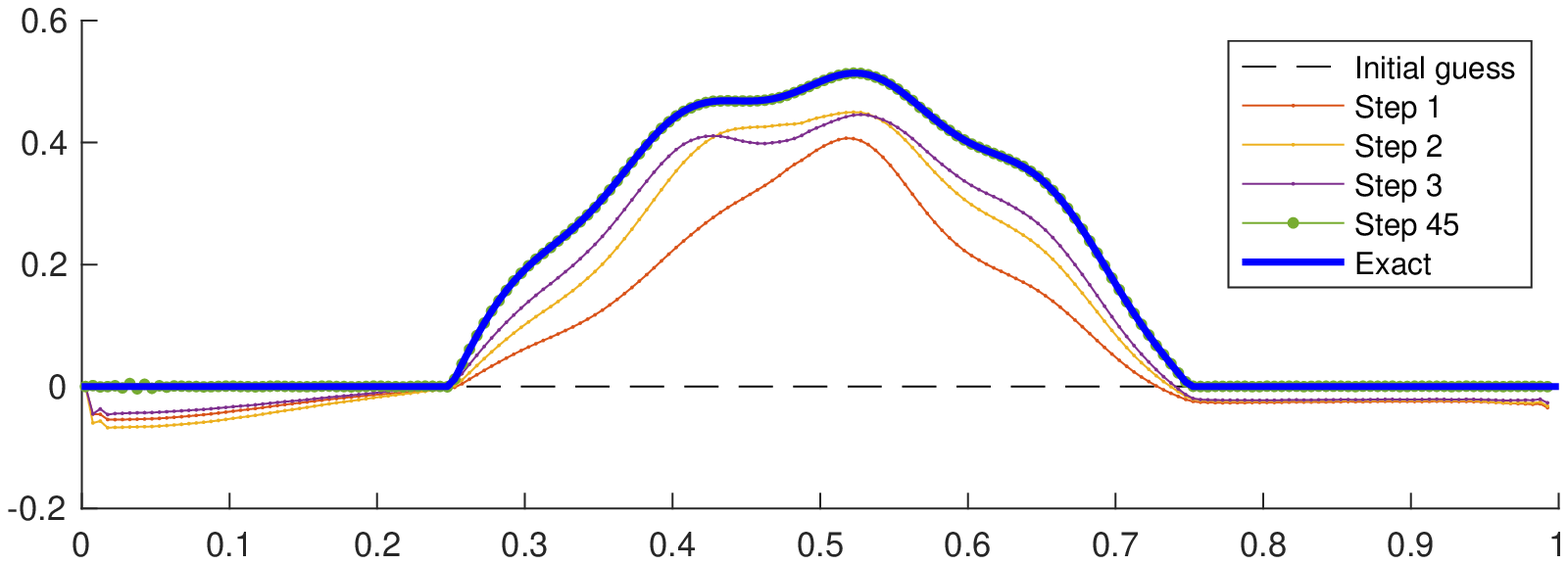}
\includegraphics[width=0.49 \textwidth]{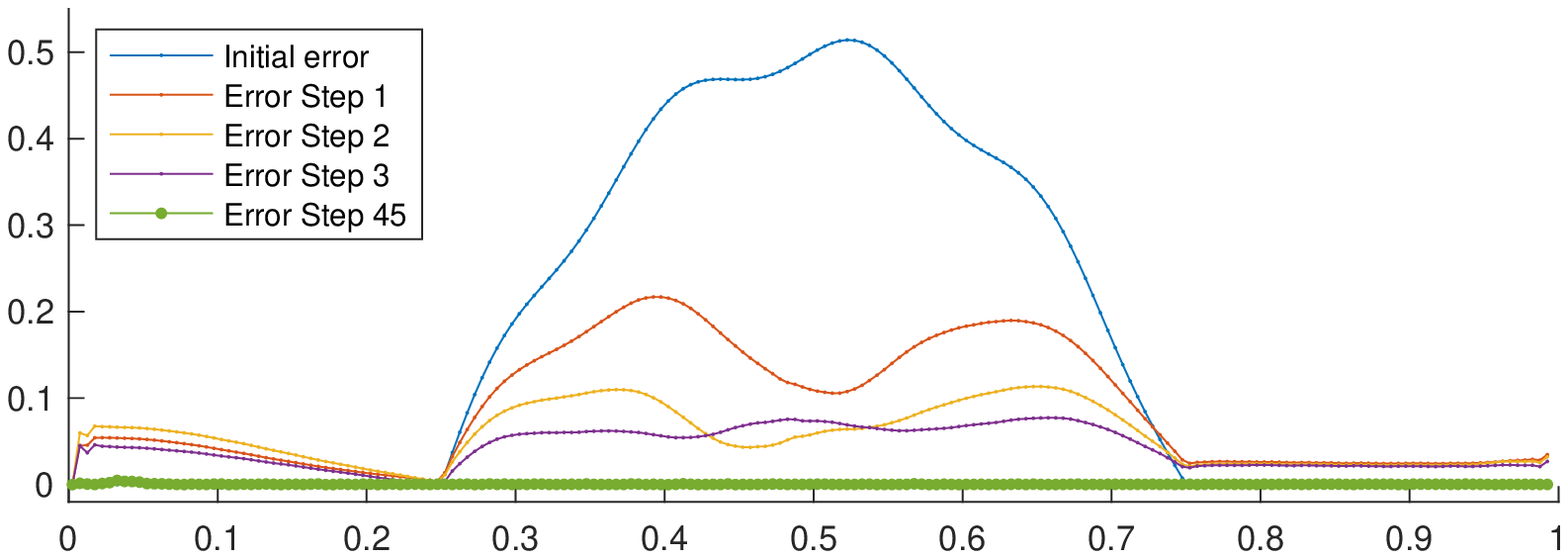}
\caption{\label{fig:BumpSinNonStat} Left panel: Exact bathymetry (blue solid line), the initial guess (black dashed line), and the intermediate steps in the algorithm (dotted lines). Right panel: Error in steps 1,2,3 and 45. }
\end{figure}
\end{center}

The ultimate goal in this work is to estimate the bathymetry in transient flows. For that end, the observed velocity in this numerical tests is time dependent and the synthetic data is constructed as follows. Given the bathymetry in equation \eqref{eq:BumpBExact}, the initial condition in the direct problem is the smooth supercritical steady state associated to the discharge $Q_{\text{steady}}=8$ and the total height $w_{\text{out}}= B_{\text{out}}+h_{\text{out}}=1$.  However, the discharge imposed at the left boundary is $Q_{\text{in}} = 9.6$, which is $20\%$ higher compared to that corresponding to the steady state. The resulting transient flow consists of a perturbation to a steady state. The right-going shockwave propagates and passes through the bump in the bathymetry. This generates a time-dependent velocity which is used as the synthetic data in the adjoint problem.

\begin{center}
\begin{figure}[h]
\centering
\includegraphics[width=0.49 \textwidth]{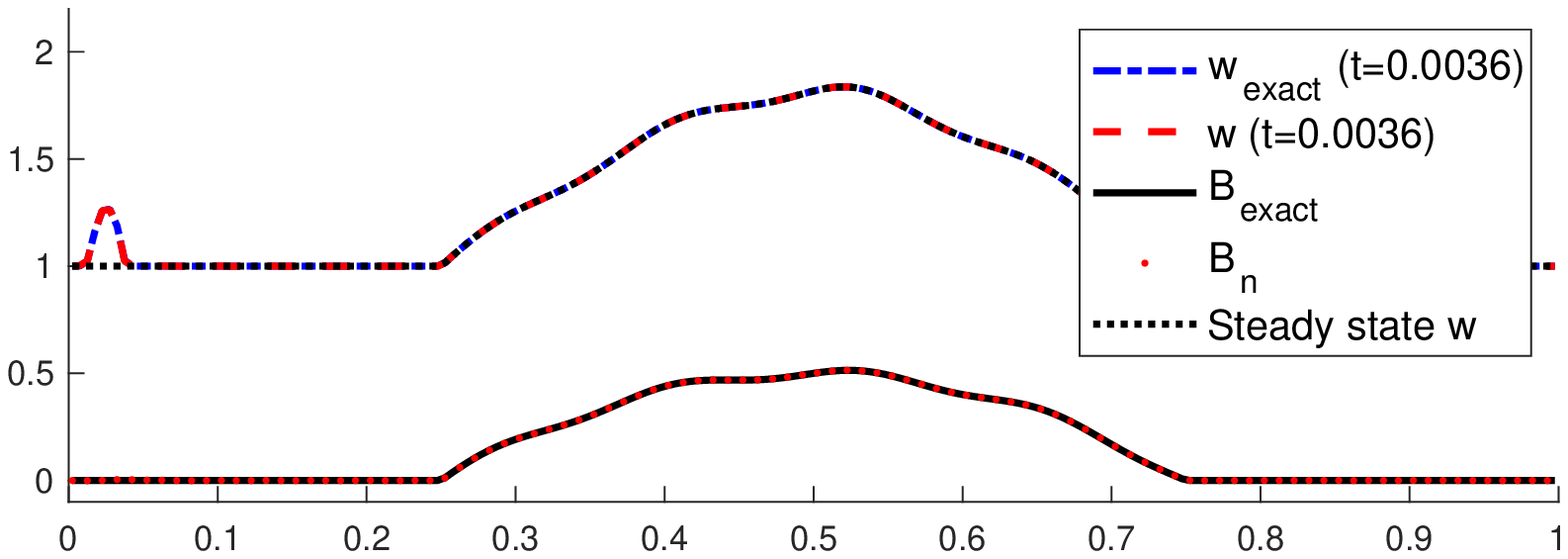}
\includegraphics[width=0.49 \textwidth]{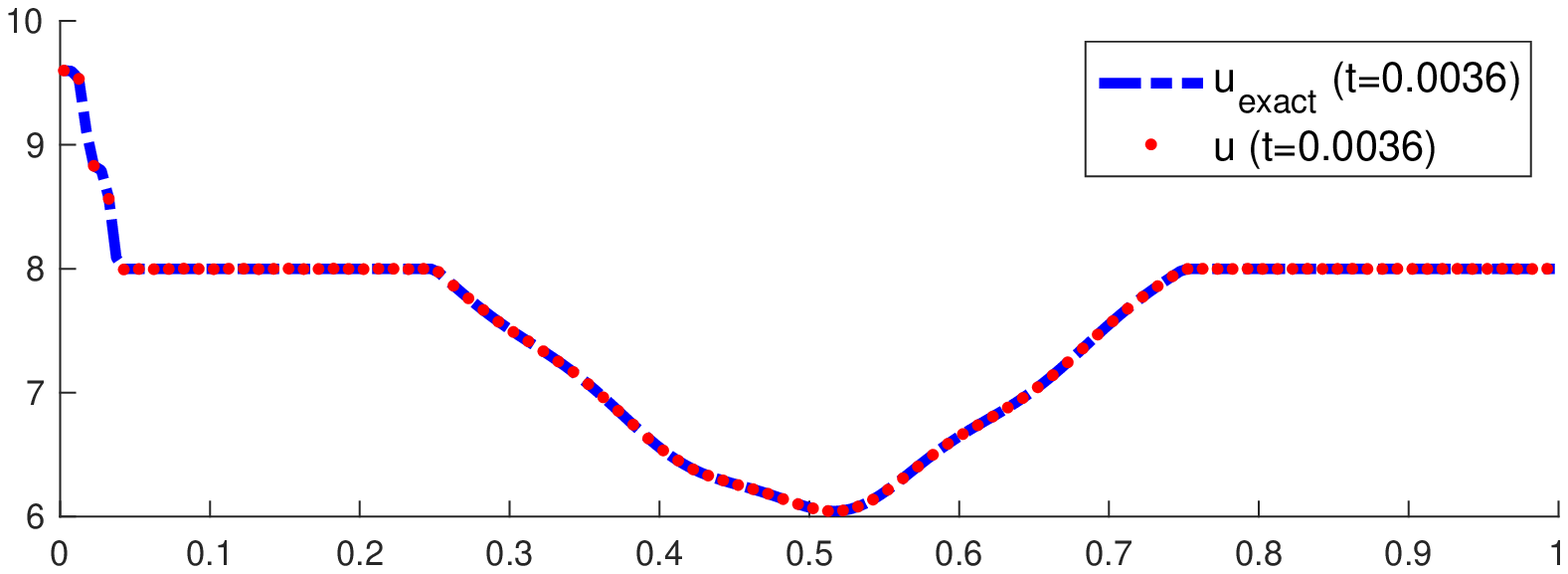}\\
\includegraphics[width=0.49 \textwidth]{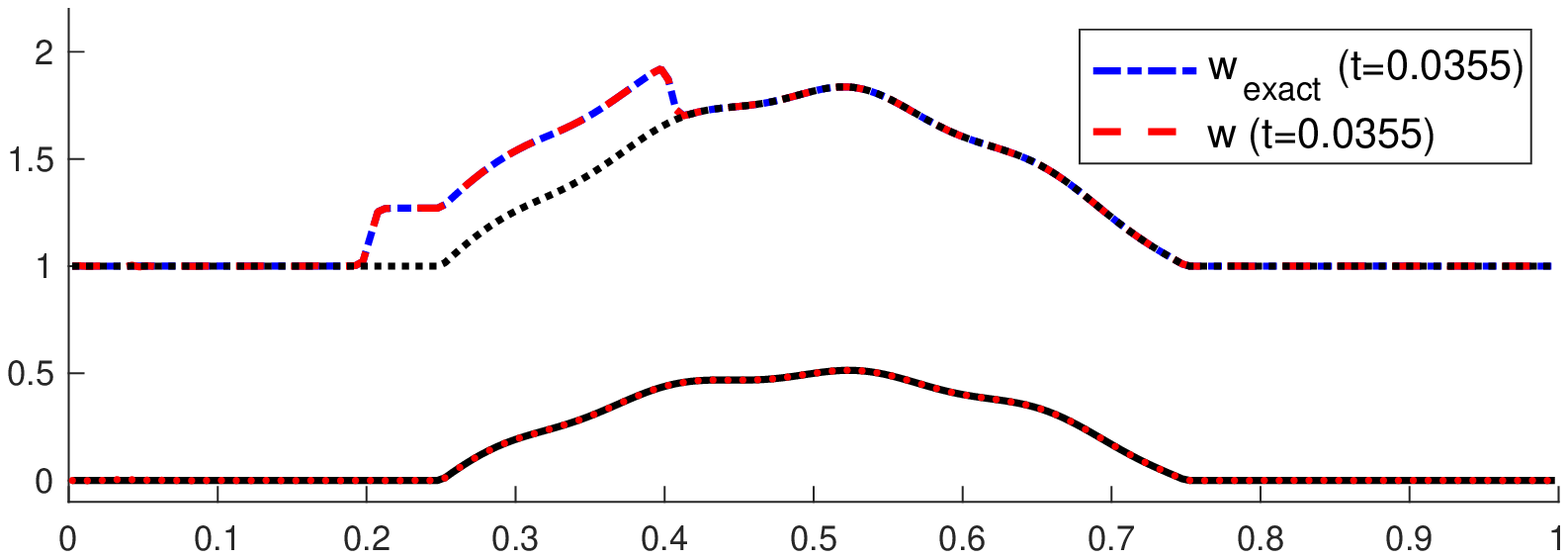}
\includegraphics[width=0.49 \textwidth]{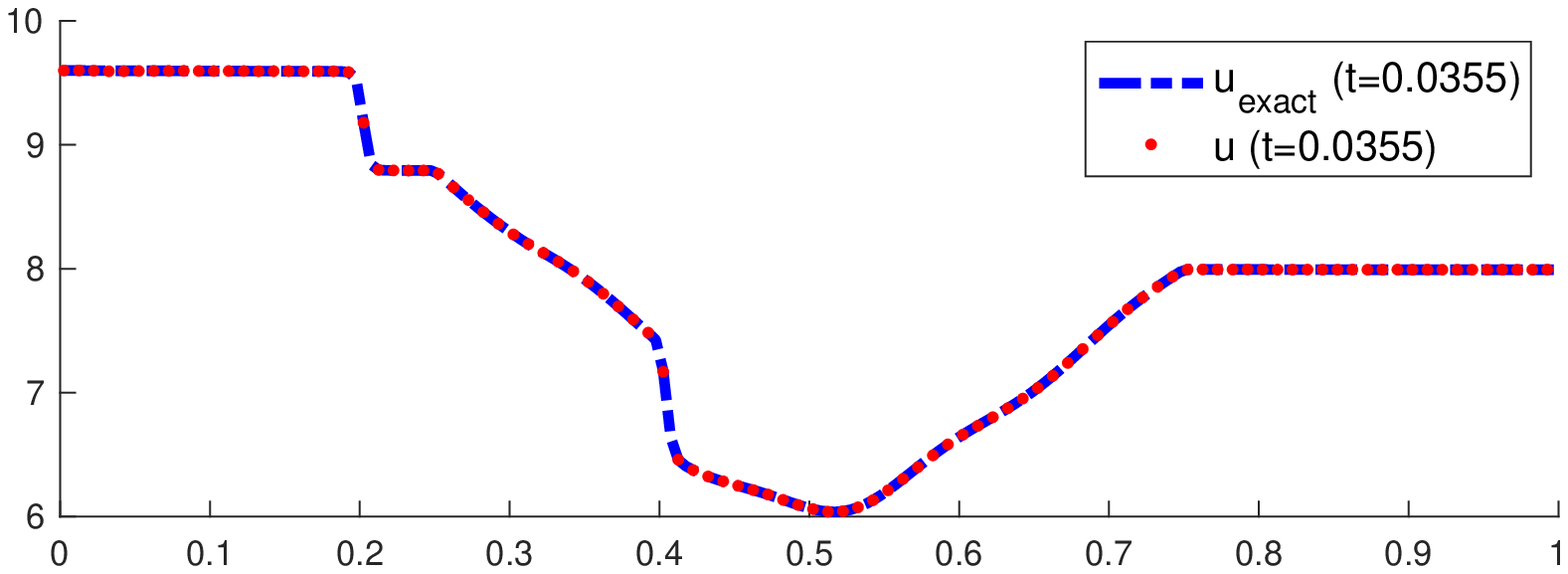}\\
\includegraphics[width=0.49 \textwidth]{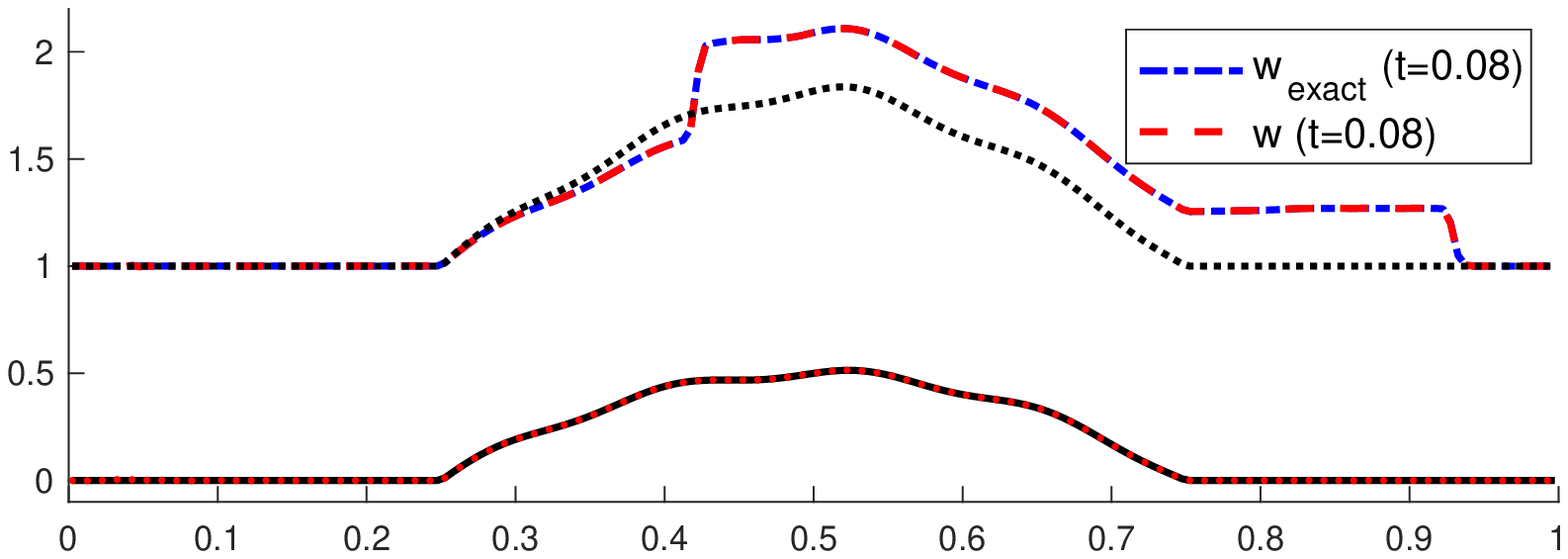}
\includegraphics[width=0.49 \textwidth]{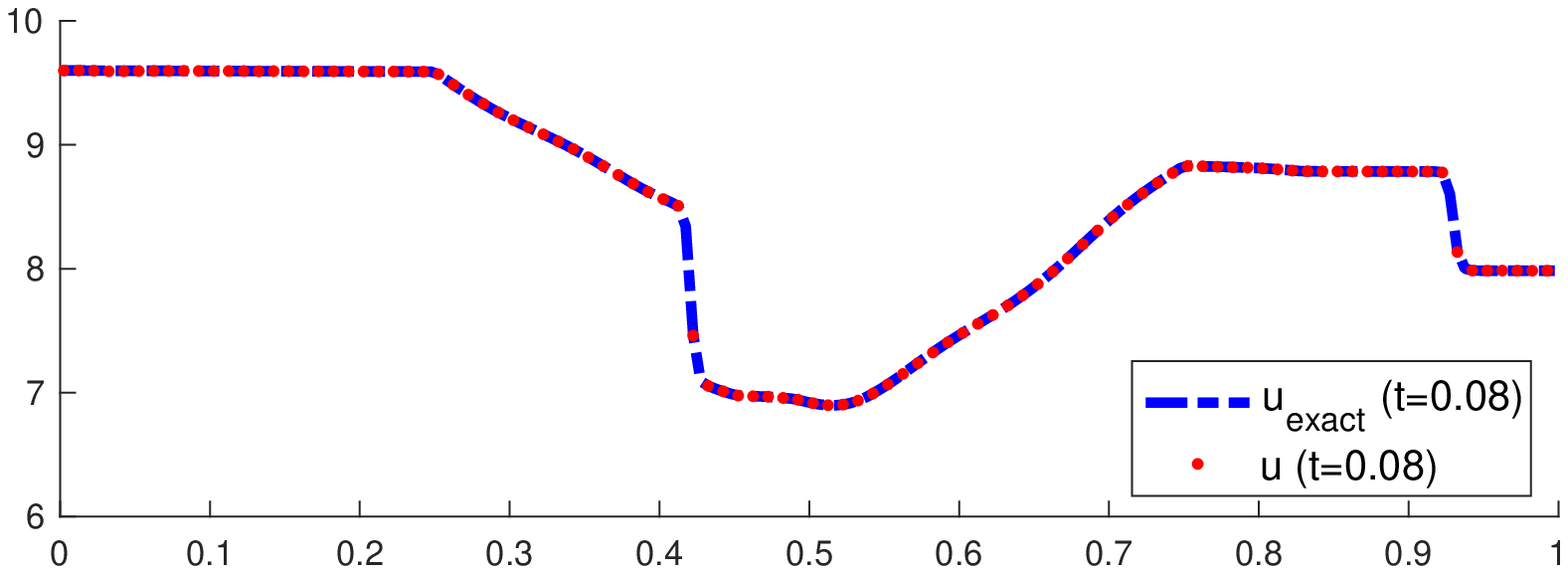}
\caption{\label{fig:TransientFlow} Left column: Approximated (blue dashed line) and exact (red dashed line) total heights at times $t=0.0036, 0.0355$, and $0.08$ in descending order, for the transient flow of Section \ref{sec:TestTransientFlow}. The exact (solid black line) and approximated (dotted red line) bathymetries are also shown. The steady state height is included to highlight the difference compared to the transient flow (dotted black line). Right column: approximated (dotted red line) and exact (dashed blue line) velocities. }
\end{figure}
\end{center}

The Manning's friction coefficient is fixed to $n=0.009 \text{ s}\text{ m}^{-1/3}$ in this numerical test and the end time $T=0.2$. Figure \ref{fig:BumpSinNonStat} shows the approximated bathymetry in steps 1,2,3, and 45 in the left panel. For completeness, the error is exhibited in the right panel.  The maximum error in the last step (away from the left boundary) is $e=8.96 \times 10^{-4}$, which corresponds to a relative error of $e_{\text{rel}} = 0.17\%$. Flows in realistic applications may not be in steady state equilibrium. Estimating the bathymetry from transient velocity measurements is a lot more challenging compared to a situation where the flow corresponds to a steady state. The present algorithm has shown to be efficient in those circumstances. 

Figure \ref{fig:TransientFlow} shows the time evolution of the total height (left column) and velocity (right column) at times  $t=0.0036, 0.03555$, and $0.08$. The total height is computed both using both the exact and the approximated bathymetries. The exact total height is denoted by the dashed blue line, while the approximated solution is identified with the dashed red line. The steady state total height is also shown for reference (black dotted line). The exact and approximated topographies are denoted by the solid blue and dotted red lines, respectively. The approximated solution is very accurate even in this time-dependent problem, and the differences are hard to be distinguished.


\subsection{Bump in a channel with varying width, friction and discontinuous top surface}
\label{sec:BumpDiscSurf}

\begin{center}
\begin{figure}[h]
\centering
\includegraphics[width=0.48 \textwidth]{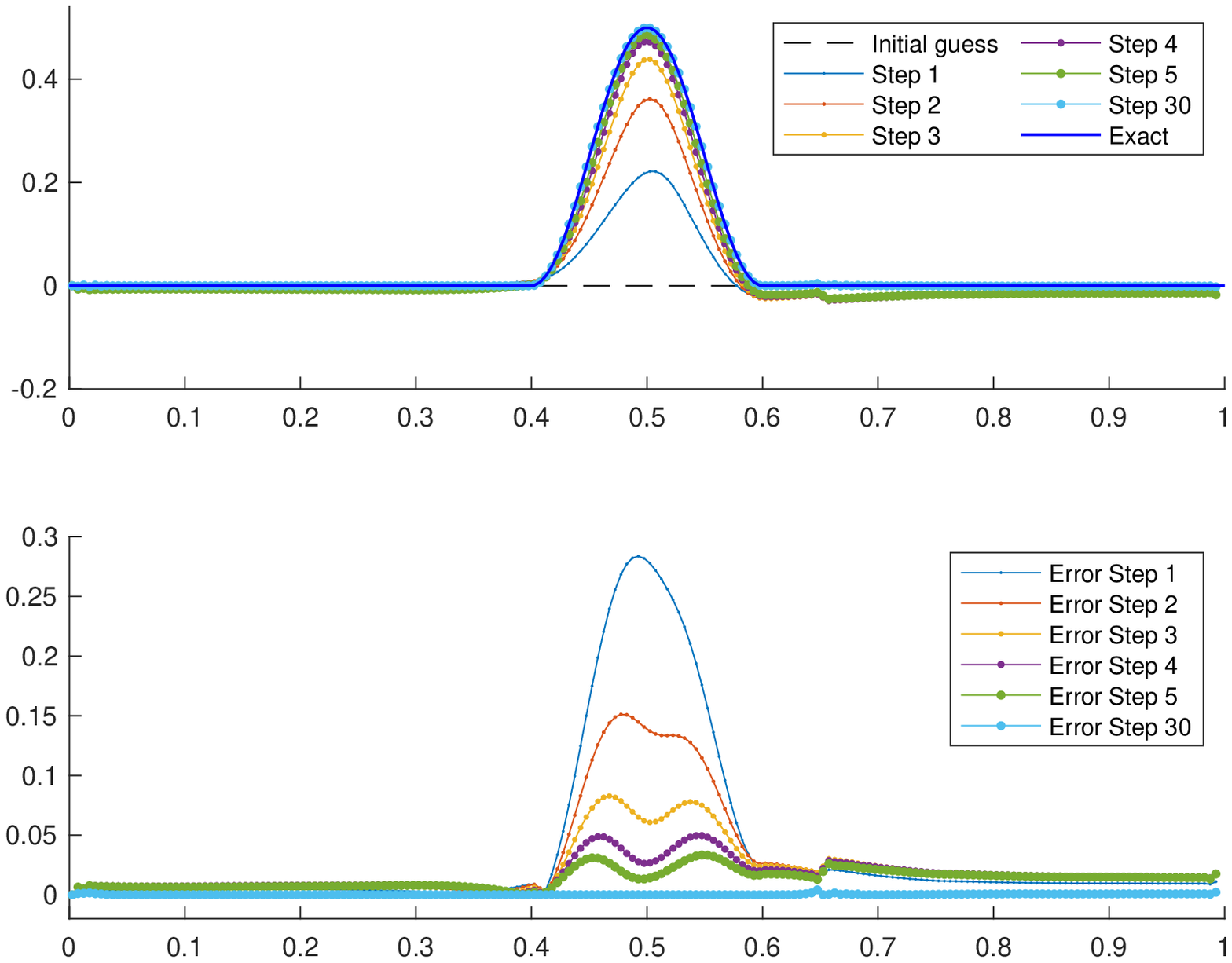}
\includegraphics[width=0.50 \textwidth]{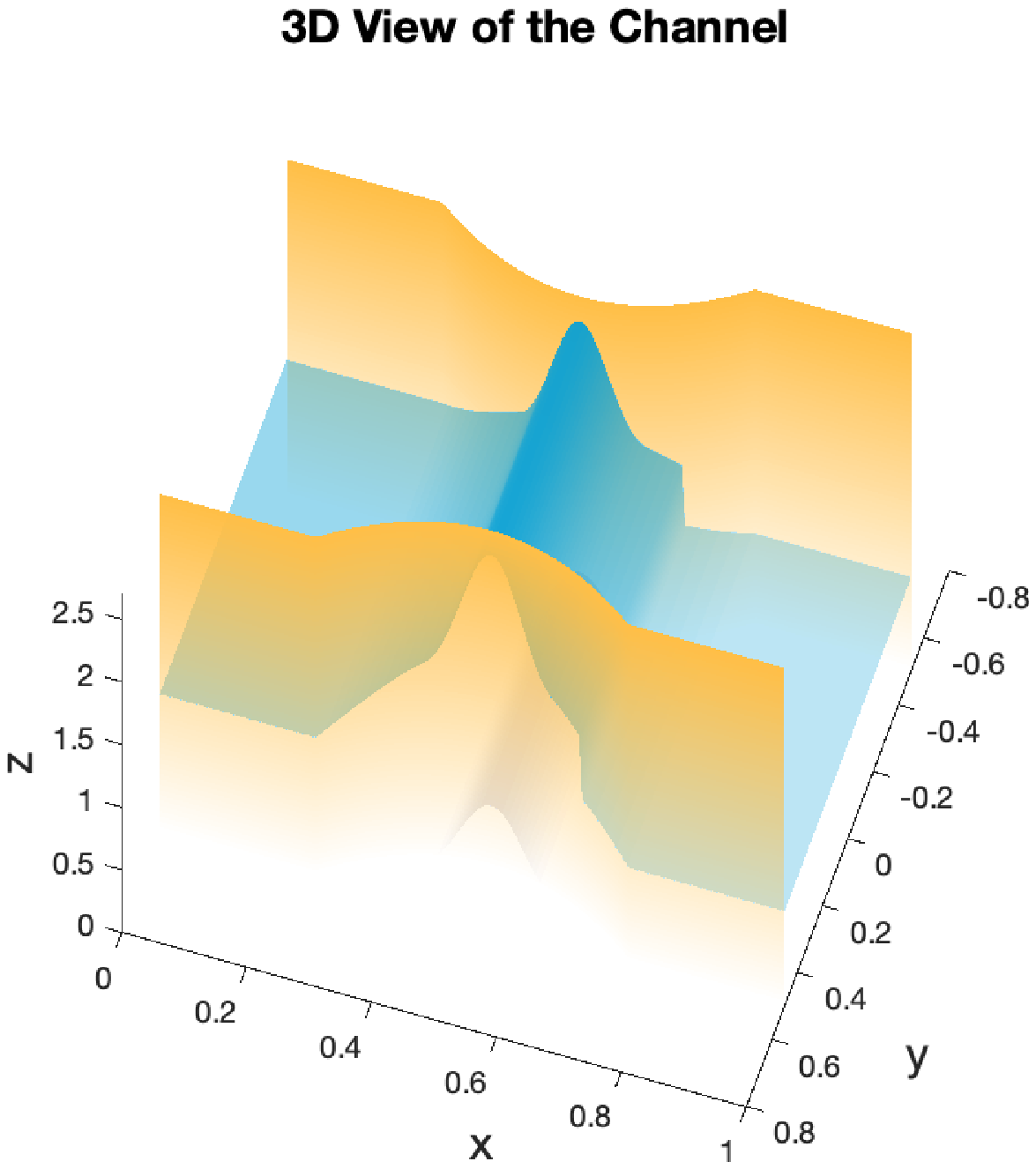}
\caption{\label{fig:InverseCodeSigmaBumpDiscSurf} Top left panel: Exact bathymetry (blue solid line), the initial guess (black dashed line), and the intermediate steps in the algorithm (dotted lines). Bottom left panel: Error in steps 1,2,3,4,5 and 30. Right panel: 3D view of the channel (yellow surface), the exact bathymetry (black surface), and the initial surface elevation $w = B+h$ (blue surface). }
\end{figure}
\end{center}

Discontinuous weak solutions of hyperbolic systems like the shallow water equations may appear in finite time. The robustness of the algorithm is tested here by estimating the bathymetry from velocity data with discontinuities. The exact bathymetry is given by
\[
B(x) =
\left\{
\begin{array}{lcl}
\frac{1}{4} \left[  \cos \left( 10 \pi \left( x-\frac{1}{2} \right) \right)  +1\right] & \text{ if } & 0.4 < x \le 0.6, \\
0 &  & \text{otherwise,}
\end{array},
\right.
\]
and the width $\sigma$ is given by equation \eqref{eq:sigma}.

The computed state in this numerical test is a steady state with shockwave in the absence of friction. However, here the Manning's friction coefficient is set to $n=0.009 \text{ s}\text{ m}^{-1/3}$, following \cite{khan2014modeling}. The initial height at the left and right boundaries are $w_{\text{in}} = 1.1, w_{\text{out}} = 0.75$. An initial shock is placed at $x_{\text{shock}}=0.65$ with left and right states given by $w_{\text{left}} = 1.417, u_{\text{left}} = 8.085,$ and $w_{\text{right}} = 0.931, u_{\text{right}} = 12.198$. The energy is initially piecewise constant, satisfying $E_{\text{in}} = 46.59$ for $x\le x_{\text{shock}}$, and $E_{\text{out}} = 83.52$ for $x > x_{\text{shock}}$. This shockwave is stationary in the absence of friction. The right panel of Figure \ref{fig:InverseCodeSigmaBumpDiscSurf} shows the 3D view of the channel (yellow surface), the bathymetry $B$ (black surface), and the initial height's elevation $w$ (blue surface). 


The top left panel of Figure \ref{fig:InverseCodeSigmaBumpDiscSurf} shows the estimated bathymetry given by the algorithm in Section \ref{sec:SearchMethod} at steps 1,2,3,4,5 and 30.  Here we use a time window $[0,T]$ with $T=0.2$. The initial state is $B_o =0$. This is significantly far from the target, which consists of a bump at the center of the domain. The first step already has a bump-like structure, with a small jump near the shockwave. At step 5, the bathymetry is close the the target and at the final step 30, the approximation is virtually on top of the exact bathymetry. The error as a function of $x$ is shown in the bottom left panel for different steps, where the convergence to the exact solution is evident. At step 30, the error is below $e=4.2 \times 10^{-3}$, which corresponds a relative error of $e_{\text{rel}} = 0.85 \%$ of the maximum bathymetry's elevation. We note that the algorithm works well even in the presence of shockwaves and friction.

\subsection{Bathymetry and Manning's friction coefficient inversion}

\begin{center}
\begin{figure}[h!]
\centering
\includegraphics[width=0.49 \textwidth]{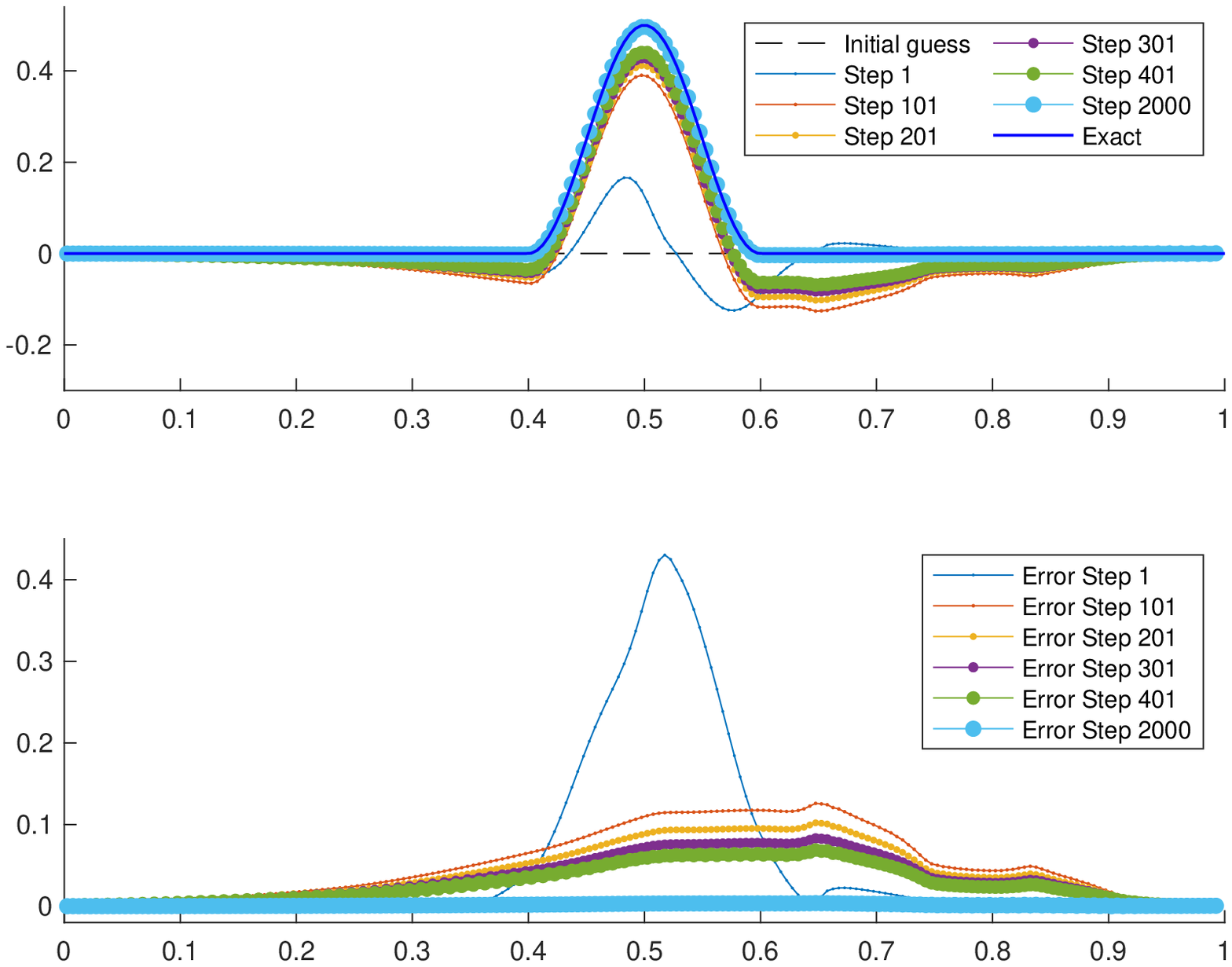}
\includegraphics[width=0.49 \textwidth]{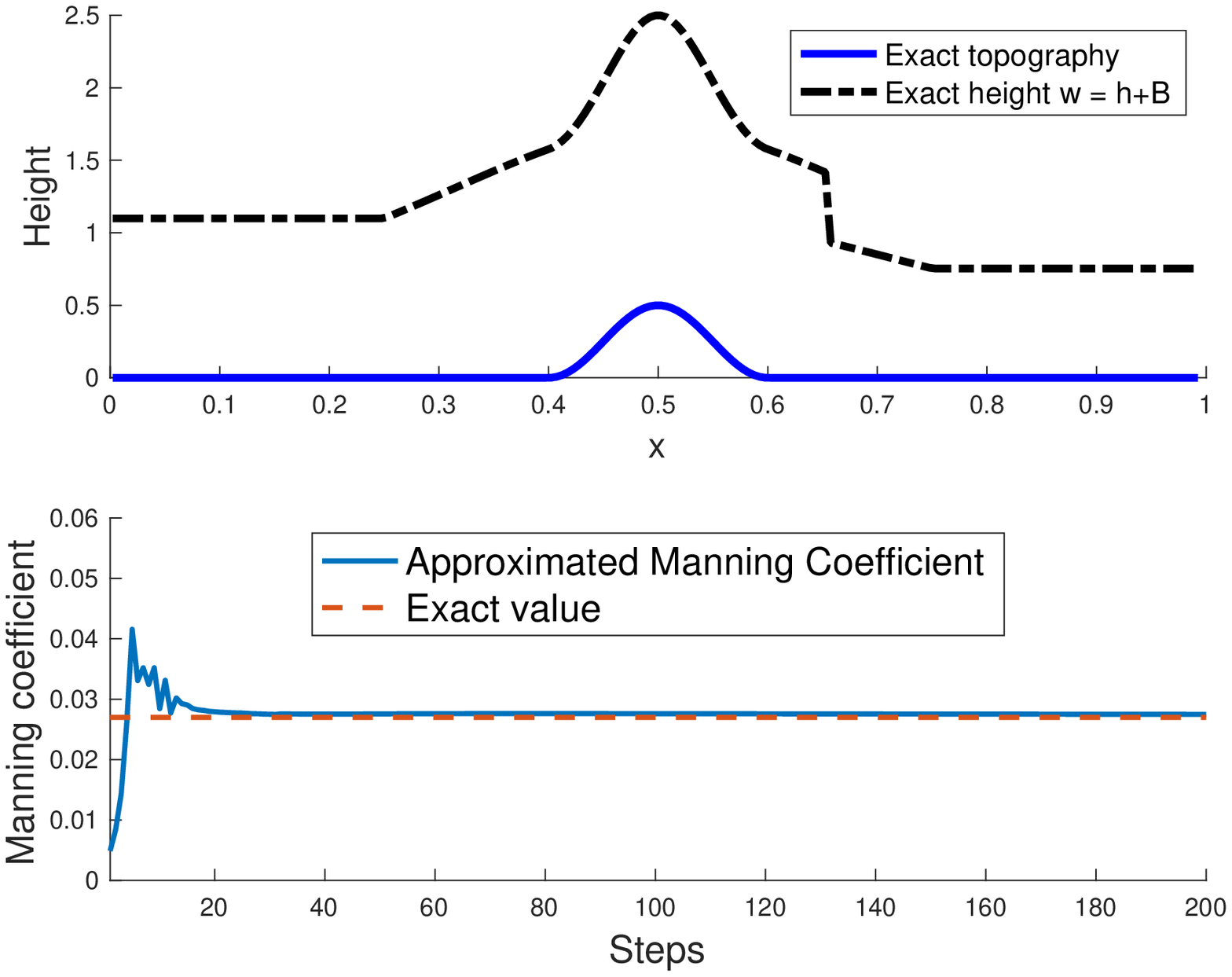}
\caption{\label{fig:InverseCodeSigmaBumpManningDiscSurf} Top left panel: Exact bathymetry (blue solid line), the initial guess (black dashed line), and the intermediate steps in the algorithm (dotted lines). Bottom left panel: Error in steps 1,101,201,301,401 and 2000. Top right panel: The exact topography (solid blue line) and the exact total height (black dashed line) are displayed. Bottom right panel: The approximated Manning's friction coefficient is shown as a function of iteration step (solid blue line), and the exact coefficient is included for reference in the dashed red line.}
\end{figure}
\end{center}

In this test, we invert both the bathymetry and the Manning's friction coefficient simultaneously. Although the value of $n$ in Section \ref{sec:BumpDiscSurf} is realistic, it may not have a strong effect in the flow in the time window considered here. As a sensitivity test, we have increased the target value of the exact Manning's friction coefficient to $n_{\text{exact}}=0.027 \text{ s}\text{ m}^{-1/3}$, which is three times larger compared to the previous one. We chose a smaller time window with $T=0.02$. However, in this case it took 2000 iteration steps to converge to the exact solution with an error of $e=3.8 \times 10^{-3}$, which corresponds to a relative error of $e_{\text{rel}} = 0.77 \%$ of the maximum bathymetry's elevation.

Using the same symbols as in Figure \ref{fig:InverseCodeSigmaBumpDiscSurf}, the top left panel of Figure \ref{fig:InverseCodeSigmaBumpManningDiscSurf} shows the exact bathymetry, and the approximated bathymetry at the initial and intermediate steps 101,201,301, 401 and 2000. The error is shown in the bottom left panel. Although it took many more steps, the error at the final step is very small. 

We can simultaneously estimate both the bathymetry's elevation and the Manning's friction coefficient in the algorithm in Section \ref{sec:SearchMethod}. The second component of the gradient $\nabla J$ has the approximated friction coefficient $n$ as a factor itself. As a result, the initial value cannot be zero because it represents an equilibrium value in the algorithm. We set the initial value of the Manning's friction coefficient as $n_o = 0.0027 = \frac{1}{10} n_{\text{exact}}$. The bottom right panel of Figure \ref{fig:InverseCodeSigmaBumpManningDiscSurf} shows the estimated Manning's friction coefficient as a function of step  number. The estimated value is already close to the exact value after about 20 steps. We only show 200 steps to see the variations in the early steps. However, the plot for 2000 steps (not shown) shows a convergence to the exact value. 

The top right panel of Figure \ref{fig:InverseCodeSigmaBumpManningDiscSurf} shows the bathymetry $B$ and the initial surface elevation with a shockwave used in the present numerical test and the previous Section \ref{sec:BumpDiscSurf}. The algorithm provides very accurate results in flows with or without friction, steady or transient states, with initial guesses that are significantly far from the exact solutions.

\subsection{Manning's coefficient inversion in the presence of wet-dry states} 
\label{sec:DamBreakReal}

The numerical test in this last section is motivated by laboratory experiments of dam breaks conducted in converging/diverging channels. See for instance, Chapter 5 of the book \cite{khan2014modeling} for a list of experiments in channels with different bed slopes and different wet and dry conditions. The experiments in \cite{khan2014modeling} Section 5.3.4 were taken from \cite{bellos1992experimental}. The channel has vertical walls and width variations along the $x$-axis, approximately given by the graph in the left panel of Figure \ref{fig:DamBreakRealCaseSigmaHeight}. The channel's length is 21.2 m, and its width is 1.4 m from 0 to 5 m, and from 16.8 to 21.2 m. The minimum width is 0.6 m at $x_m = 8.5 \text{ m}$.  

 \begin{figure}[h!]
\begin{center}
{\includegraphics[width=0.39 \textwidth]{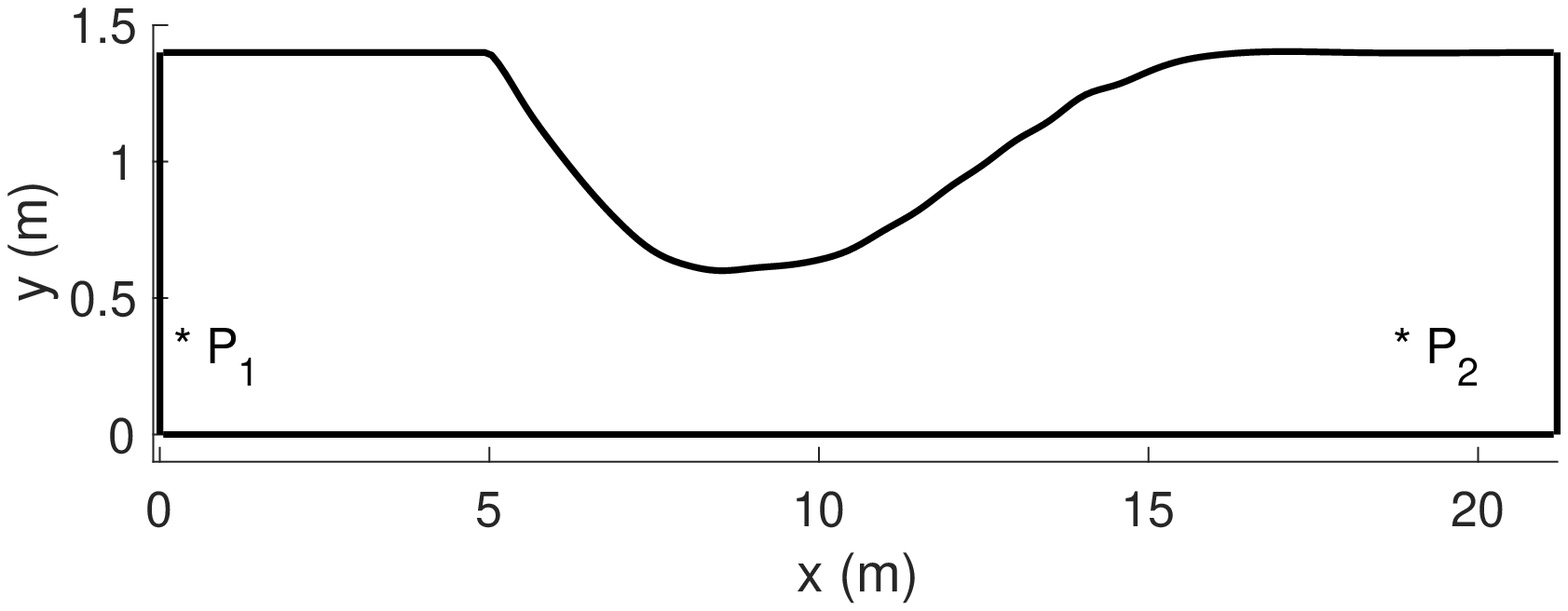}}
{\includegraphics[width=0.59 \textwidth]{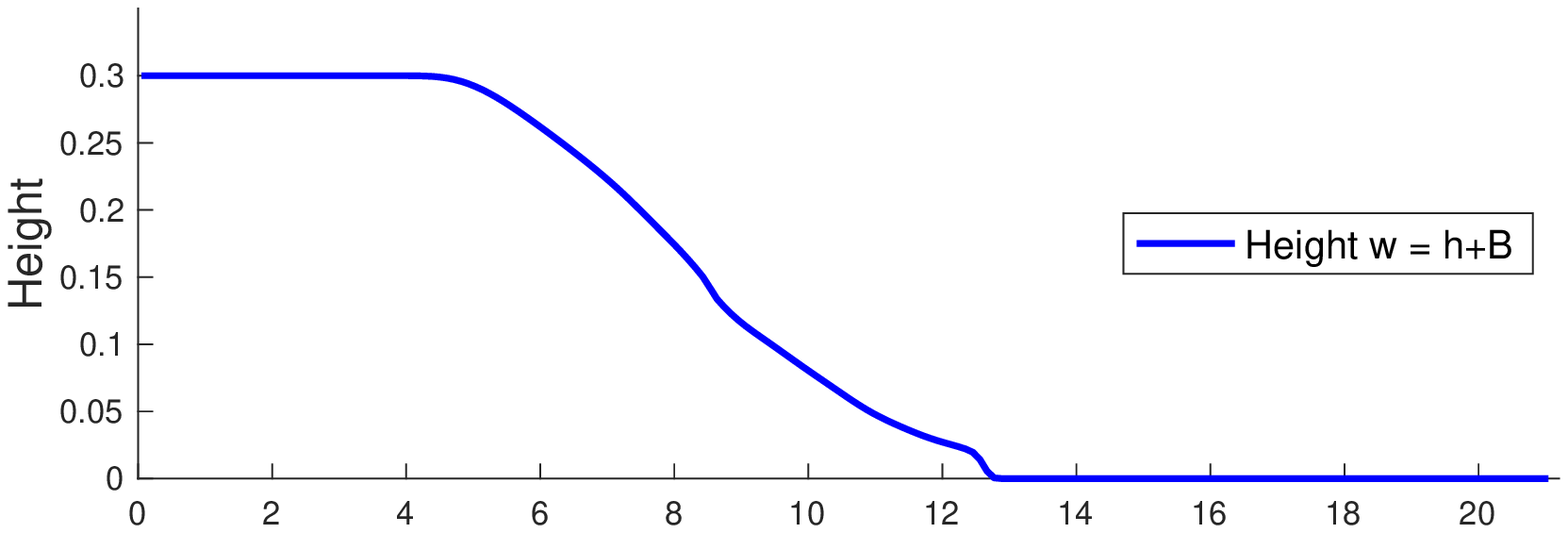}}
\end{center}
\caption{\label{fig:DamBreakRealCaseSigmaHeight} Left panel: Approximated channel's width as a function of $x$ (top view of the channel). The points $P_1$ and $P_2$ indicate the locations where the depth was measured in the corresponding experiment in \cite{khan2014modeling}. Right panel: Total height $w=h+B$ at time $t=4\text{ s}$ with initial conditions \eqref{eq:InitCondRealCase} and the boundary conditions below. }
\end{figure}

In the experiment, the flow is initially given by
\begin{equation}
\label{eq:InitCondRealCase}
u(x,t=0) = 0, w(x,t=0) = 
\left\{
\begin{array}{lcl}
0.3 \text{m} & \text{ if } & x < x_m = 8.5 \text{ m},\\ 
H_{\text{out }}= 10^{-5} \text{ m}& & \text{ otherwise },
\end{array}
\right. 
\end{equation}
which corresponds to a flow initially at rest, where the downstream part of the channel is dry (a threshold value has been used). The gate is assumed to be instantaneously removed. The left boundary is a solid wall. We have used zero Dirichlet left boundary conditions in the velocity and Neumann left boundary conditions for the height. The right boundary extrapolates the data at outflow, and imposes $H_{\text{out}}$ at inflow. Once the dam breaks, the flow evolves as illustrated in  the left panel of Figure \ref{fig:DamBreakRealCaseComp} at $t=4   \text{ s}$. The resolution here is $\Delta x = 21.2 \text{ m}/200$.

The bathymetry is flat $B=0$. So, we are interested in estimating the Manning's friction coefficient, which is unknown. Unfortunately, the depth at two locations ($P_1$ and $P_2$ in Figure \ref{fig:DamBreakRealCaseComp}) are the only quantities reported in this experiment. In \cite{hernandez2016central}, it was found that one good approximation for the Manning's friction coefficient is $n=0.0084 \text{ s}\text{ m}^{-1/3}$. Here we create synthetic velocity data based on this value to approximate the friction coefficient based on those velocity measurements. The purpose of this numerical test is to show that the algorithm works well even in the presence of wet-dry states, in connection with the above experiment. 

 \begin{figure}[h!]
\begin{center}
{\includegraphics[width=0.49 \textwidth]{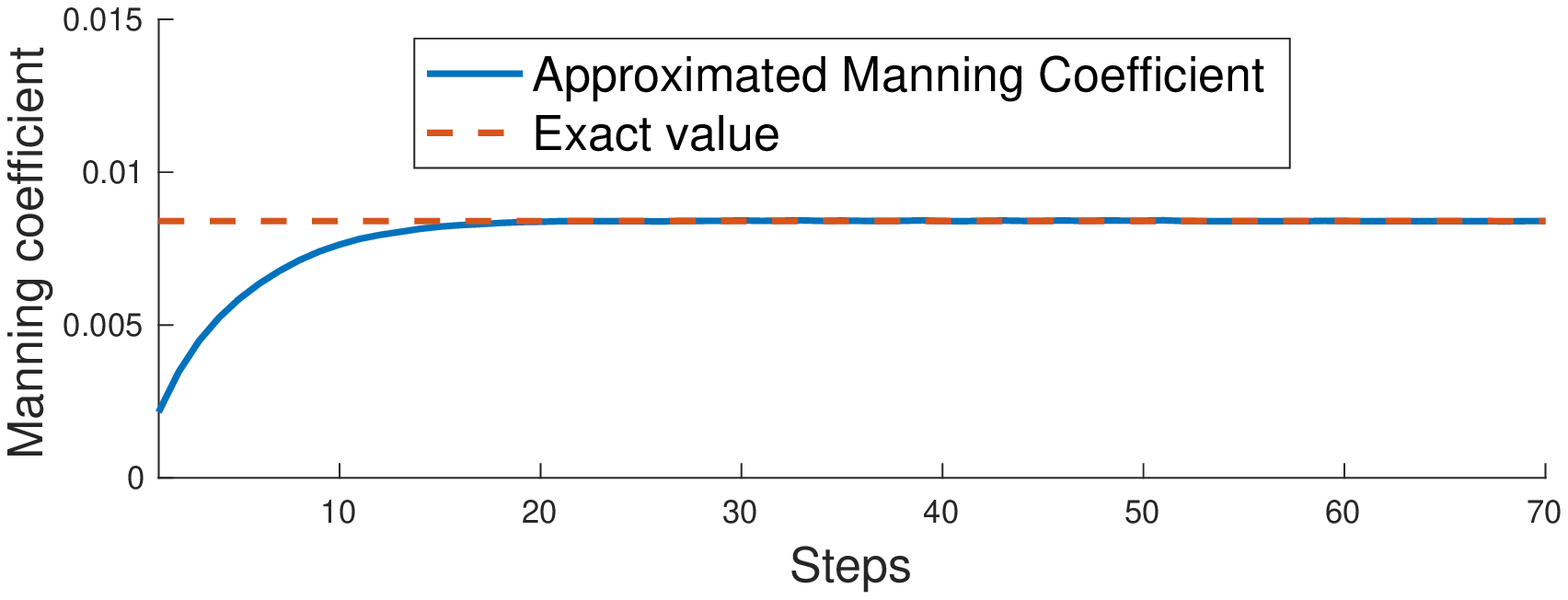}}
{\includegraphics[width=0.49 \textwidth]{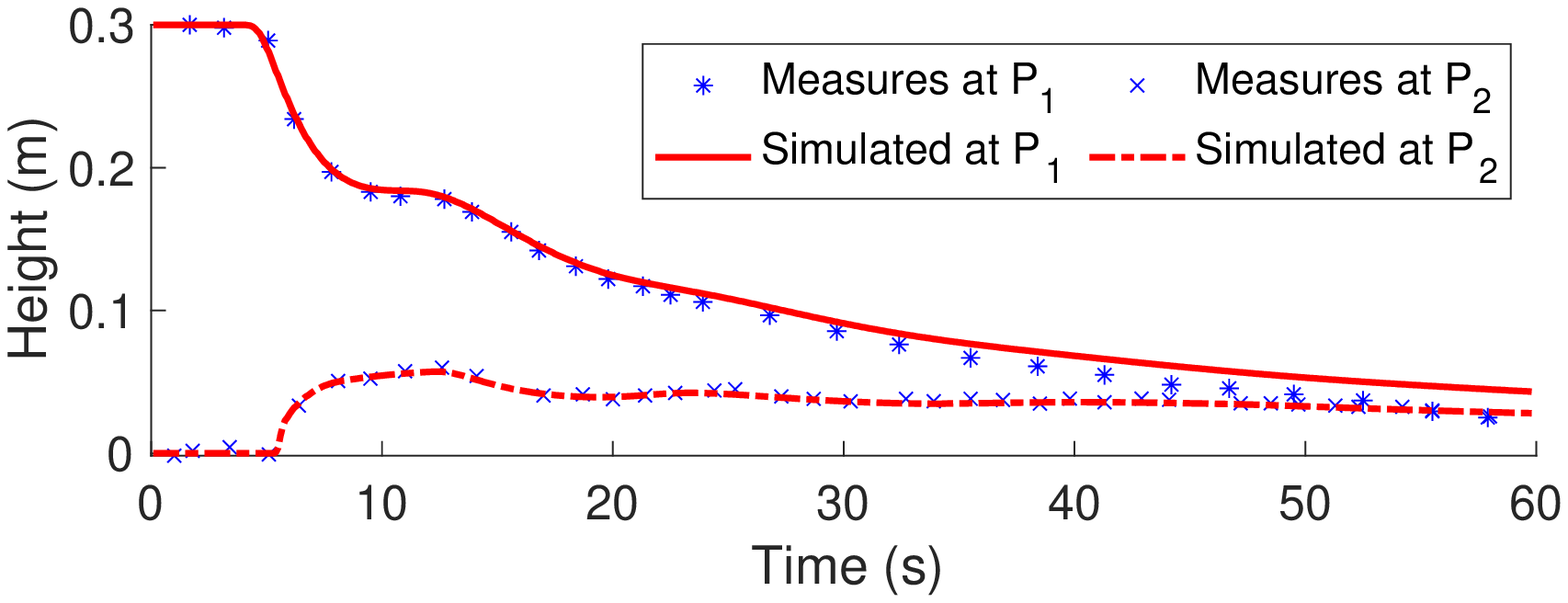}}
\end{center}
\caption{\label{fig:DamBreakRealCaseComp} Left panel: Convergence of the Manning's friction coefficient to the correct value. The approximations are plotted for different steps. Right panel: comparison between experimental data and the numerical approximation obtained by the present schemes of the water's height at two particular locations in $x$, versus time. The left point is located at the left boundary $P_1 = 0$, and the right point is located at $P_2 = 18.5 \text{m}$.}
\end{figure}

The approximated Manning's coefficient are shown for different steps in the left panel of Figure \ref{fig:DamBreakRealCaseComp}. One can observe that the Manning's coefficient is very closed to the exact value after 20 steps in the algorithm. The approximated values were obtained using a time window $[0,T]$ with $T=4 \text{ s}$ before the flow reaches the boundaries. 

In the experimental data in \cite{bellos1992experimental,khan2014modeling}, the height was measured in time at two particular locations. One at the left boundary $P_1 = 0$, and the other one near the right boundary $P_2 = 18.5 \text{ m}$. The right panel in Figure \ref{fig:DamBreakRealCaseComp} compares the real and numerical values. We observe a good agreement, specially at the location $P_2$ near the right boundary, for the entire simulation.  The numerical approximation at $P_1$ is accurate for the first part of the simulation, and overestimates it for the second half. Boundary conditions and the adjustment of the Manning coefficient might affect the predictions. 

\section{Conclusions}

In this work we formulated a constrained optimization problem to estimate the bed channel bathymetry and Manning's friction coefficient from available data of the fluid's velocity. The continuous problem is first presented and analyzed. A quadratic functional which by means of Lagrange multipliers incorporates the shallow water equations is minimized using the Fr\'echet derivative. A continuous descent method is formulated to obtain the minimal solution. Both direct and adjoint systems are proposed to be solved by a second-order Roe-type upwind numerical scheme. However, the algorithm works for any other efficient and robust numerical scheme. We estimate the bathymetry of transient flows as well as the Manning's friction coefficient. Several benchmark problems are presented in order to verify the numerical performance of the proposed method. A  simple steady-state case is first formulated to verify the reliability of the algorithm before transient flows can be treated. In this first case we considered a steady state velocity and a bathymetry bump with a sinusoidal perturbation. In a second test we considered a transient flow consisting of a right-going perturbation to a steady state. Finally, we simultaneously estimated both the bathymetry and the Manning's friction coefficient in a channel with varying width and discontinuous top surface and a numerical test was presented to estimate the Manning's coefficient in the presence of wet-dry states, motivated by experimental data. We obtained very accurate approximations of the bathymetry in all cases. 

We have provided an algorithm that works very well even in transient flows in channels with vertical walls of varying width, discontinuous top surfaces, and even wet-dry states. The need for a good initial guess and the empirical initial coefficient ($\alpha_k$) in the search direction are often limitations for approaches like the one presented here. However, we have shown that our algorithm is not very sensitive to those parameters and we provided criteria to choose the best coefficient $\alpha_k$ together with conditions to stop our algorithm. Furthermore, our setting is flexible and may be adapted to estimate other parameters or systems.  \\


\noindent
{\bf Acknowledgements:}

Research supported in part by grants UNAM-DGAPA-PAPIIT IN113019 \& Conacyt A1-S-17634.

\appendix

\section{Appendix. Gradient by the Adjoint State Method}
\label{sec:AppendixGradient}

In this appendix, we compute an expression for the Fr\'{e}chet derivative of $J$.

\begin{lemma} Let  $h(B,n)$, $u(B,n)$ solve the shallow water equations for given $(B,n)$. Let $H=Dh(B,n)(\xi_1,\xi_2)$ and $U=Du(B,n)(\xi_1,\xi_2)$. Then 
\begin{equation}
 \begin{array}{lcl}
DJ(B,n)(\xi_1,\xi_2)  & = &   \langle\mu,g\sigma h(\xi_1)_x\rangle\,
+ \langle\mu,2g\frac{\sigma h}{R^{4/3}}u\sqrt{u^2+\varepsilon}\rangle\xi_2 
  + \langle\lambda,(\sigma H)_t+(\sigma uH )_x\rangle \, \\
 & & +  \langle\mu,(\sigma uH)_t+(\sigma u^2H +g\sigma hH)_x\rangle \, +  \langle\mu,-gh\sigma_xH+g\sigma B_xH\rangle \, \\
 & & + \langle\mu,gn^2\left(\frac{1}{h}+\frac{2}{\sigma}\right)^{1/3}\left(2-\frac{1}{3}\frac{\sigma}{h}\right) u\sqrt{u^2+\varepsilon}\, H\rangle \\
  & & + \langle U,\mathcal{M}^*(\mathcal{M}u-\hat{u}) \, 
+ \langle\lambda,(\sigma h U)_x\rangle \\
& & + \langle\mu,(\sigma h U)_t+2(\sigma huU)_x+gn^2\frac{\sigma h}{R^{4/3}}\frac{2u^2+\varepsilon}{\sqrt{u^2+\varepsilon}}U\rangle.
 \end{array}
 \end{equation}
\end{lemma}

\begin{proof}
First, we note that
\[
\begin{array}{lcl}
D\mathcal{L}(B,n,h,u,\lambda,\mu)(\xi_1,\xi_2,\xi_3,\xi_4)  
& = &  D_1\mathcal{L}(B,n,h,u,\lambda,\mu)\xi_1\, 
+  D_2\mathcal{L}(B,n,h,u,\lambda,\mu)\xi_2 +\\

 & &  D_3\mathcal{L}(B,n,h,u,\lambda,\mu)\xi_3\, +
 D_4\mathcal{L}(B,n,h,u,\lambda,\mu)\xi_4.
 \end{array}
 \]
 But,
 \[
D_1\mathcal{L}(B,k,h,u,\lambda,\mu)\xi_1 =   \langle\mu,g\sigma h(\xi_1)_x\rangle ,
 \]
 \[
  D_2\mathcal{L}(B,k,h,u,\lambda,\mu)\xi_2 = \langle\mu,2g\frac{\sigma h}{R^{4/3}}u\sqrt{u^2+\varepsilon}\rangle\xi_2,
 \]
 \[
 \begin{array}{lcl}
D_3\mathcal{L}(B,k,h,u,\lambda,\mu)\xi_3  & = & \langle\lambda,(\sigma \xi_3)_t+(\sigma u\xi_3 )_x\rangle \, + 
  \langle\mu,(\sigma u\xi_3)_t+(\sigma u^2\xi_3 +g\sigma h\xi_3)_x\rangle \, \\
 & & + \langle\mu,-gh\sigma_x\xi_3+g\sigma B_x\xi_3\rangle \, +
  \langle\mu,gn^2\left(\frac{1}{h}+\frac{2}{\sigma}\right)^{1/3}\left(2-\frac{1}{3}\frac{\sigma}{h}\right) u\sqrt{u^2+\varepsilon}\,\xi_3\rangle,
\end{array}
 \]
\[
D_4\mathcal{L}(B,k,h,u,\lambda,\mu)\xi_4 =  
\langle\xi_4,\mathcal{M}^*(\mathcal{M}u-\hat{u}) \, 
+ \langle\lambda,(\sigma h\xi_4)_x\rangle 
 + \langle\mu,(\sigma h\xi_4)_t+2(\sigma hu\xi_4)_x+gn^2\frac{\sigma h}{R^{4/3}}\frac{2u^2+\varepsilon}{\sqrt{u^2+\varepsilon}}\xi_4\rangle .
 \]

Let $W(B,n)=(B,n,h(B,n),u(B,n))$. Then
\[
J(B,n)=\mathcal{L}(W(B,n)).
\]

By the chain rule,
\begin{equation}
\begin{array}{lcl}
DJ(B,n)(\xi_1,\xi_2) & = & D\mathcal{L}(W(B,n))DW(B,n)(\xi_1,\xi_2) \\
& = & D\mathcal{L}(W(B,n))(\xi_1,\xi_2,Dh(B,n)(\xi_1,\xi_2),Du(B,n)(\xi_1,\xi_2)) \\
& \equiv &   D\mathcal{L}(W(B,n))(\xi_1,\xi_2,H,U).
\end{array}
 \end{equation}
 
 This leads to
\begin{equation}
 \label{DJ_B}
 \begin{array}{lcl}
DJ(B,k)(\xi_1,\xi_2)  & = & \langle\mu,g\sigma h(\xi_1)_x\rangle\,
+ \langle\mu,2g\frac{\sigma h}{R^{4/3}}u\sqrt{u^2+\varepsilon}\rangle\xi_2  + \langle\lambda,(\sigma H)_t+(\sigma uH )_x\rangle \,  \\
 & & + \langle\mu,(\sigma uH)_t+(\sigma u^2H +g\sigma hH)_x\rangle \, +  \langle\mu,-gh\sigma_xH+g\sigma B_xH\rangle \, \\
 & & + \langle\mu,gn^2\left(\frac{1}{h}+\frac{2}{\sigma}\right)^{1/3}\left(2-\frac{1}{3}\frac{\sigma}{h}\right) u\sqrt{u^2+\varepsilon}\, H\rangle \\
  & & + \langle U,\mathcal{M}^*(\mathcal{M}u-\hat{u}) \, + \langle\lambda,(\sigma h U)_x\rangle \\
& & + \langle\mu,(\sigma h U)_t+2(\sigma huU)_x+gn^2\frac{\sigma h}{R^{4/3}}\frac{2u^2+\varepsilon}{\sqrt{u^2+\varepsilon}}U\rangle.
 \end{array}
\end{equation}
as required.
\end{proof}

We now proceed to prove Theorem \ref{th:Lagrange}.

\begin{proof}\textbf{(Theorem \ref{th:Lagrange})}

Since Fr\'{e}chet differentiability implies Gateaux differentiability, we have
\[
Dh(B,n)(\xi_1,\xi_2)=lim_{\varepsilon\to 0}\frac{h((B,n)+\varepsilon(\xi_1,\xi_2))-h(B,n)}{\varepsilon},
\]
and similarly
\[
Du(B,n)(\xi_1,\xi_2)=lim_{\varepsilon\to 0}\frac{u((B,n)+\varepsilon(\xi_1,\xi_2))-u(B,n)}{\varepsilon}.
\]

It follows that the functions $H\equiv Dh(B,n)(\xi_1,\xi_2)$ and $U\equiv Du(B,n)(\xi_1,\xi_2)$, are zero where  initial and boundary data are given. Namely
\[
H(x,0)=0=U(x,0),\quad H(a,t)=0=U(a,t).
\]

Let us assume that $\xi_1\in C^\infty_c(\Omega)$. Integrating by parts in (\ref{DJ_B}), we obtain
\[
DJ(B,n)(\xi_1,\xi_2) =  \langle\xi_1,-\int_0^T(g\sigma h\mu)_x\, dt\rangle\,
+ \langle\mu,2g\frac{\sigma h}{R^{4/3}}u\sqrt{u^2+\varepsilon}\rangle\xi_2  + L_1(H) + L_2(U) + BdS + BdT,
\]
where

\begin{dmath}
 L_1(H)  =  \langle H,-\sigma \lambda_t-\sigma u\lambda_x\rangle \, +
   \langle H,-\sigma u\mu_t-(\sigma u^2 +g\sigma h)\mu_x\rangle \, +
  \langle H,-gh\sigma_x\mu+g\sigma B_x\mu\rangle \, + 
 \langle H,gn^2\left(\frac{1}{h}+\frac{2}{\sigma}\right)^{1/3}\left(2-\frac{1}{3}\frac{\sigma}{h}\right) u\sqrt{u^2+\varepsilon}\, \mu\rangle ,
\end{dmath}

\begin{dmath}
 L_2(U)  = \langle U,\mathcal{M}^*(\mathcal{M}u-\hat{u}) \, 
+ \langle U,-\sigma h \lambda_x\rangle 
 + \langle U,-\sigma h \mu_t-2\sigma hu\mu_x+gn^2\frac{\sigma h}{R^{4/3}}\frac{2u^2+\varepsilon}{\sqrt{u^2+\varepsilon}}\mu\rangle ,
\end{dmath}

\begin{equation}
\begin{array}{lcl}
BdS & = & 
\left(\left.\xi_1\int_0^T\left(\mu g\sigma h\right) dt\right)\right\vert_a^b+
\int_0^T\left. \sigma h\left(\lambda +2\mu u\right)U\right\vert_a^b dt \, 
+ \int_0^T\left.\sigma\left[\lambda u +\mu(u^2+gh)\right]H\right\vert_a^b dt, 
\end{array}
\end{equation}
and 
\begin{equation}
BdT =
 \int_a^b\left.\sigma\left[\mu hU+(\lambda+\mu u )H \right]\right\vert_0^T dx.
\end{equation}

Since $h$ and $u$ are given at $t=0$ and $x=a$, we set the adjoint variables $\lambda$ and $\mu$
as null at $t=T$ and $x=b$. Namely
\[
\lambda(x,T)=\mu(x,T)=0,\quad x\in (a,b),
\]
\[
\lambda(b,t)=\mu(b,t)=0, \quad t\in (0,T).
\]

Since $\xi_1\in C^\infty_c(\Omega)$,  we obtain
\[
BdS  = 0 = BdT,
\]
and requiring all terms involving $H$ and $U$ to be null, we are led to the adjoint equations. 
Since  $C^\infty_c(\Omega)$ is dense in $L^2(\Omega)$ we obtain the gradient as required.
\end{proof}


\bibliography{References.bib}

\begin{thebibliography}{}

\bibitem[Balb{\'a}s and Karni, 2009]{balbas2009central}
Balb{\'a}s, J. and Karni, S. (2009).
\newblock A central scheme for shallow water flows along channels with
  irregular geometry.
\newblock {\em ESAIM: Mathematical Modelling and Numerical Analysis},
  43(2):333--351.

\bibitem[Belanger and Vincent, 2005]{P2Languer2005}
Belanger, E. and Vincent, A. (2005).
\newblock Data assimilation (4d-var) to forecast flood in shallow-waters with
  sediment erosion.
\newblock {\em Journal of Hydrology}, 300(1):114 -- 125.

\bibitem[Bellos et~al., 1992]{bellos1992experimental}
Bellos, C., Soulis, V., and Sakkas, J. (1992).
\newblock Experimental investigation of two-dimensional dam-break induced
  flows.
\newblock {\em Journal of Hydraulic Research}, 30(1):47--63.

\bibitem[Bolognesi et~al., 2017]{bolognesi2017measurement}
Bolognesi, M., Farina, G., Alvisi, S., Franchini, M., Pellegrinelli, A., and
  Russo, P. (2017).
\newblock Measurement of surface velocity in open channels using a lightweight
  remotely piloted aircraft system.
\newblock {\em Geomatics, Natural Hazards and Risk}, 8(1):73--86.

\bibitem[Brisset et~al., 2018]{P4Brisset2018}
Brisset, P., Monnier, J., Garambois, P.-A., and Roux, H. (2018).
\newblock On the assimilation of altimetric data in 1d saint-venant river flow
  models.
\newblock {\em Advances in Water Resources}, 119:41 -- 59.

\bibitem[Ding et~al., 2004]{ding2004identification}
Ding, Y., Jia, Y., and Wang, S.~S. (2004).
\newblock Identification of manning's roughness coefficients in shallow water
  flows.
\newblock {\em Journal of Hydraulic Engineering}, 130(6):501--510.

\bibitem[Ding and Wang, 2005]{ding2005identification}
Ding, Y. and Wang, S.~S. (2005).
\newblock Identification of manning's roughness coefficients in channel network
  using adjoint analysis.
\newblock {\em International Journal of Computational Fluid Dynamics},
  19(1):3--13.

\bibitem[Ding and Wang, 2012a]{ding2012optimal}
Ding, Y. and Wang, S.~S. (2012a).
\newblock Optimal control of flood diversion in watershed using nonlinear
  optimization.
\newblock {\em Advances in water resources}, 44:30--48.

\bibitem[Ding and Wang, 2012b]{Ding}
Ding, Y. and Wang, S.~S. (2012b).
\newblock Optimal control of flood diversion in watershed using nonlinear
  optimization.
\newblock {\em Advances in Water Resources}, 44:30 -- 48.

\bibitem[Garambois and Monnier, 2015]{GARAMBOIS2015103}
Garambois, P.-A. and Monnier, J. (2015).
\newblock Inference of effective river properties from remotely sensed
  observations of water surface.
\newblock {\em Advances in Water Resources}, 79:103 -- 120.

\bibitem[Garcia-Navarro and Vazquez-Cendon, 2000]{garcia2000numerical}
Garcia-Navarro, P. and Vazquez-Cendon, M.~E. (2000).
\newblock On numerical treatment of the source terms in the shallow water
  equations.
\newblock {\em Computers \& Fluids}, 29(8):951--979.

\bibitem[Gessese et~al., 2011]{Gessese_etal_2011}
Gessese, A., Sellier, M., Van~Houten, E., and Smart, G. (2011).
\newblock Reconstruction of river bed topography from free surface data using a
  direct numerical approach in one-dimensional shallow water flow.
\newblock {\em Inverse Problems}, 27(2):025001.

\bibitem[Gessese et~al., 2013]{Gessese_etal_2013}
Gessese, A., Smart, G., Heining, C., and Sellier, M. (2013).
\newblock One-dimensional bathymetry based on velocity measurements.
\newblock {\em Inverse Problems in Science and Engineering}, 21(4):704--720.

\bibitem[Hernandez-Duenas and Beljadid, 2016]{hernandez2016central}
Hernandez-Duenas, G. and Beljadid, A. (2016).
\newblock A central-upwind scheme with artificial viscosity for shallow-water
  flows in channels.
\newblock {\em Advances in water resources}, 96:323--338.

\bibitem[Honnorat et~al., 2007]{Honnorata}
Honnorat, M., Marin, J., Monnier, J., and Lai, X. (2007).
\newblock {Dassflow v1.0: a variational data assimilation software for 2D river
  flows}.
\newblock Research Report RR-6150, {INRIA}.

\bibitem[Honnorat et~al., 2009]{Honnoratb}
Honnorat, M., Monnier, J., and Le~Dimet, F.-X. (2009).
\newblock Lagrangian data assimilation for river hydraulics simulations.
\newblock {\em Computing and Visualization in Science}, 12:235--246.

\bibitem[Honnorat et~al., 2010]{Honnoratc}
Honnorat, M., Monnier, J., RIVIERE, N., Huot, E., and Le~Dimet, F.-X. (2010).
\newblock {Identification of equivalent topography in an open channel flow
  using Lagrangian data assimilation}.
\newblock {\em {Computing and Visualization in Science}}, 13(3):111--119.

\bibitem[Hubbard and Garcia-Navarro, 2000]{hubbard2000flux}
Hubbard, M.~E. and Garcia-Navarro, P. (2000).
\newblock Flux difference splitting and the balancing of source terms and flux
  gradients.
\newblock {\em Journal of Computational Physics}, 165(1):89--125.

\bibitem[Kawahara~M, 1990]{Kawahara}
Kawahara~M, K.~T. (1990).
\newblock A flood control of dam reservoir by conjugate gradient and finite
  element method.
\newblock The Proceedings: Fifth International Conference on Numerical Ship
  Hydrodynamics, pages 18--29. The National Academies Press.

\bibitem[Kevlahan et~al., 2019]{Kevlahan_etal}
Kevlahan, N.-R., Khan, R., and Protas, B. (2019).
\newblock On the convergence of data assimilation for the one-dimensional
  shallow water equations with sparse observations.
\newblock {\em Advances in Computational Mathematics}, 45(5):3195--3216.

\bibitem[Khan and Lai, 2014]{khan2014modeling}
Khan, A.~A. and Lai, W. (2014).
\newblock {\em Modeling shallow water flows using the discontinuous Galerkin
  method}.
\newblock CRC Press.

\bibitem[Kurganov et~al., 2007]{kurganov2007second}
Kurganov, A., Petrova, G., et~al. (2007).
\newblock A second-order well-balanced positivity preserving central-upwind
  scheme for the saint-venant system.
\newblock {\em Communications in Mathematical Sciences}, 5(1):133--160.

\bibitem[Lacasta et~al., 2017]{Lacasta_etal}
Lacasta, A., Morales-Hern{\'a}ndez, M., Burguete, J., Brufau, P., and
  Garc{\'\i}a-Navarro, P. (2017).
\newblock Calibration of the 1d shallow water equations: a comparison of monte
  carlo and gradient-based optimization methods.
\newblock {\em Journal of Hydroinformatics}, 19(2):282--298.

\bibitem[Le~Dimet et~al., 2003]{Dimet2003}
Le~Dimet, F.-X., Navon, I.~M., and Daescu, D.~N. (2003).
\newblock {Second-Order Information in Data Assimilation*}.
\newblock {\em Monthly Weather Review}, 130(3):629--648.

\bibitem[Lee et~al., 2018]{Lee_etal}
Lee, J., Ghorbanidehno, H., Farthing, M.~W., Hesser, T.~J., Darve, E.~F., and
  Kitanidis, P.~K. (2018).
\newblock Riverine bathymetry imaging with indirect observations.
\newblock {\em Water Resources Research}, 54(5):3704--3727.

\bibitem[Leveque, 1992]{leveque1992numerical}
Leveque, R.~J. (1992).
\newblock {\em Numerical methods for conservation laws}, volume~3.
\newblock Springer.

\bibitem[LeVeque et~al., 2011]{leveque2011tsunami}
LeVeque, R.~J., George, D.~L., and Berger, M.~J. (2011).
\newblock Tsunami modelling with adaptively refined finite volume methods.
\newblock {\em Acta Numerica}, 20:211--289.

\bibitem[Marks and Bates, 2000]{Marks_K}
Marks, K. and Bates, P. (2000).
\newblock Integration of high-resolution topographic data with floodplain flow
  models.
\newblock {\em Hydrological Processes}, 14(11-12):2109--2122.

\bibitem[Mazauric, 2003]{Mazauric2003}
Mazauric, C. (2003).
\newblock {\em Assimilation de donn\'es pour les mod\'eles d’hydraulique
  fluviale. Estimation de param\'etres, analyse de sensibilit\'e et
  d\'ecomposition. Mod\'elisation et simulation.}
\newblock PhD thesis, Grenoble I.

\bibitem[Mazauric et~al., 2004]{MAZAURIC2004403}
Mazauric, C., Tran, V., Castaings, W., Froehlich, D., and {Le Dimet}, F.
  (2004).
\newblock Parallel algorithms and data assimilation for hydraulic models.
\newblock In Joubert, G., Nagel, W., Peters, F., and Walter, W., editors, {\em
  Parallel Computing}, volume~13 of {\em Advances in Parallel Computing}, pages
  403 -- 411. North-Holland.

\bibitem[Montecinos et~al., 2019]{montecinos2019numerical}
Montecinos, G.~I., L{\'o}pez-R{\'\i}os, J.~C., Ortega, J.~H., and Lecaros, R.
  (2019).
\newblock A numerical procedure and coupled system formulation for the adjoint
  approach in hyperbolic pde-constrained optimization problems.
\newblock {\em IMA Journal of Applied Mathematics}, 84(3):483--516.

\bibitem[Nguyen et~al., 2014]{nguyen2014optimal}
Nguyen, V.~T., Georges, D., and Besancon, G. (2014).
\newblock Optimal state estimation in an overland flow model using the adjoint
  method.
\newblock In {\em 2014 IEEE Conference on Control Applications (CCA)}, pages
  2034--2039. IEEE.

\bibitem[Nguyen et~al., 2016a]{Nguyen_etal_2016}
Nguyen, V.~T., Georges, D., and Besan{\c{c}}on, G. (2016a).
\newblock State and parameter estimation in 1-d hyperbolic pdes based on an
  adjoint method.
\newblock {\em Automatica}, 67:185--191.

\bibitem[Nguyen et~al., 2016b]{nguyen2016parameter}
Nguyen, V.~T., Georges, D., Besan{\c{c}}on, G., and Zin, I. (2016b).
\newblock Parameter estimation of a real hydrological system using an adjoint
  method.
\newblock {\em IFAC-PapersOnLine}, 49(13):300--305.

\bibitem[Roe, 1987]{roe1987upwind}
Roe, P. (1987).
\newblock Upwind differencing schemes for hyperbolic conservation laws with
  source terms.
\newblock In {\em Nonlinear hyperbolic problems}, pages 41--51. Springer.

\bibitem[Roe, 1981]{roe1981approximate}
Roe, P.~L. (1981).
\newblock Approximate riemann solvers, parameter vectors, and difference
  schemes.
\newblock {\em Journal of computational physics}, 43(2):357--372.

\bibitem[Sanders and Katopodes, 1999]{Sandersa}
Sanders, B.~F. and Katopodes, N.~D. (1999).
\newblock Control of canal flow by adjoint sensitivity method.
\newblock {\em Journal of Irrigation and Drainage Engineering},
  125(5):287--297.

\bibitem[Sanders~BF, 1999]{Sandersb}
Sanders~BF, K.~N. (1999).
\newblock Active flood hazard mitigation, part 2: Omnidirectional wave control.
\newblock {\em ASCE J Hydraulic Eng}, 125(10):1071--83.

\bibitem[Sanders~BF, 2000]{Sandersc}
Sanders~BF, K.~N. (2000).
\newblock Adjoint sensitivity analysis for shallow-water wave control.
\newblock {\em J Eng Mech}, 126(9):909--19.

\bibitem[V{\'a}zquez-Cend{\'o}n, 1999]{vazquez1999improved}
V{\'a}zquez-Cend{\'o}n, M.~E. (1999).
\newblock Improved treatment of source terms in upwind schemes for the shallow
  water equations in channels with irregular geometry.
\newblock {\em Journal of Computational Physics}, 148(2):497--526.

\bibitem[Westaway et~al., 2001]{Westaway}
Westaway, R.~M., Lane, S.~N., and Hicks, D.~M. (2001).
\newblock Remote sensing of clear-water, shallow, gravel-bed rivers using
  digital photogrammetry.
\newblock {\em Photogrammetric Engineering and Remote Sensing},
  67(11):1271--1282.

\end{thebibliography}

\end{document}